\documentclass[hidelinks]{eptcs}
\usepackage{enumerate}
\usepackage{breakurl}
\usepackage{underscore}
\usepackage{latexsym}
\usepackage{wrapfig}
\usepackage{marvosym}
\usepackage{color}
\usepackage{graphicx}
\usepackage{amssymb}
\usepackage{stmaryrd}
\usepackage{multirow}
\providecommand{\publicationstatus}{}
\advance\textheight 13.6pt

\DeclareSymbolFont{frenchscript}{OMS}{ztmcm}{m}{n}
\DeclareMathSymbol{\A}{\mathord}{frenchscript}{65}    % set of CCS handshake communications
\DeclareMathSymbol{\Ch}{\mathord}{frenchscript}{67}   % set of CCS handshake communications
\DeclareMathSymbol{\Pow}{\mathord}{frenchscript}{80}  % powerset
\DeclareMathSymbol{\V}{\mathord}{frenchscript}{88}    % set of variables
\DeclareMathAlphabet{\mathcal}{OMS}{cmsy}{m}{n}       % latex mathcal default
\DeclareMathAlphabet{\mathbbm}{U}{bbm}{m}{n}          % blackboard bold
\newcommand{\IE}{\mathbbm{E}}                         % expressions
\newcommand{\IT}{\mathbbm{T}}                         % tests
\newcommand{\IR}{\mathbbm{R}}                         % real numbers
\newtheorem{defi}{Definition}
\newtheorem{theo}{Theorem}
\newtheorem{prop}{Proposition}
\newtheorem{lemm}{Lemma}
\newtheorem{coro}{Corollary}
\newtheorem{obs}{Observation}
\newtheorem{exam}{Example}
\newtheorem{op}{Question}

\newenvironment{definition}[1]{\begin{defi} \rm \label{df:#1} }{\end{defi}}
\newenvironment{definitionA}[2]{\begin{defi}[#1] \rm \label{df:#2} }{\end{defi}}
\newenvironment{theorem}[1]{\begin{theo} \rm \label{thm:#1} }{\end{theo}}
\newenvironment{proposition}[1]{\begin{prop} \rm \label{pr:#1} }{\end{prop}}

\newenvironment{example}[1]{\begin{exam} \rm \label{ex:#1} }{\end{exam}}

\newenvironment{open}[1]{\begin{op} \rm \label{op:#1} }{\end{op}}
\newenvironment{proof}{\begin{trivlist} \item[\hspace{\labelsep}\bf Proof:]}{\hfill $\Box$\end{trivlist}}

\newcommand{\Sec}[1]{Section~\ref{sec:#1}}
\newcommand{\Sects}{Sections}

\newcommand{\df}[1]{Definition~\ref{df:#1}}
\newcommand{\thm}[1]{Theorem~\ref{thm:#1}}
\newcommand{\pr}[1]{Proposition~\ref{pr:#1}}

\newcommand{\ex}[1]{Example~\ref{ex:#1}}
\newcommand{\tab}[1]{Table~\ref{tab:#1}}
\newcommand{\fig}[1]{Figure~\ref{#1}}

\newenvironment{itemise}%
  {\begin{itemize}%
    \setlength{\itemsep}{0pt}%
    \setlength{\parskip}{0pt}}%
  {\end{itemize}}
\makeatletter
\def\comesfrom{\@transition\leftarrowfill}
\def\goesto{\@transition\rightarrowfill}
\def\ngoesto{\@transition\nrightarrowfill}
\def\Goesto{\@transition\Rightarrowfill}
\def\nGoesto{\@transition\nRightarrowfill}
\def\xmapsto{\@transition\mapstofill}
\def\nxmapsto{\@transition\nmapstofill}
\def\@transition#1{\@@transition{#1}}
\newbox\@transbox
\newbox\@arrowbox
\newbox\@downbox
\def\@@transition#1#2%
   {\setbox\@transbox\hbox
      {\vrule height 1.5ex depth .8ex width 0ex\hskip0.25em$\scriptstyle#2$\hskip0.25em}
   \ifdim\wd\@transbox<1.5em
      \setbox\@transbox\hbox to 1.5em{\hfil\box\@transbox\hfil}\fi
   \setbox\@arrowbox\hbox to \wd\@transbox{#1}
   \ht\@arrowbox\z@\dp\@arrowbox\z@
   \setbox\@transbox\hbox{$\mathop{\box\@arrowbox}\limits^{\box\@transbox}$}
   \dp\@transbox\z@\ht\@transbox 10pt
   \mathrel{\box\@transbox}}
\def\nrightarrowfill{$\m@th\mathord-\mkern-6mu%
  \cleaders\hbox{$\mkern-2mu\mathord-\mkern-2mu$}\hfill
  \mkern-6mu\mathord\not\mkern-2mu\mathord\rightarrow$}
\def\Rightarrowfill{$\m@th\mathord=\mkern-6mu%
  \cleaders\hbox{$\mkern-2mu\mathord=\mkern-2mu$}\hfill
  \mkern-6mu\mathord\Rightarrow$}
\def\nRightarrowfill{$\m@th\mathord=\mkern-6mu%
  \cleaders\hbox{$\mkern-2mu\mathord=\mkern-2mu$}\hfill
  \mkern-6mu\mathord\not\mathord\Rightarrow$}
\def\mapstofill{$\m@th\mathord\mapstochar\mathord-\mkern-6mu%
  \cleaders\hbox{$\mkern-2mu\mathord-\mkern-2mu$}\hfill
  \mkern-6mu\mathord\rightarrow$}
\def\nmapstofill{$\m@th\mathord\mapstochar\mathord-\mkern-6mu%
  \cleaders\hbox{$\mkern-2mu\mathord-\mkern-2mu$}\hfill
  \mkern-6mu\mathord\not\mkern-2mu\mathord\rightarrow$}
\makeatother %%%%% end of arrows definition
\newcommand{\ar}[1]{\mathrel{\goesto{#1}}}            % arrow
\newcommand{\goto}[1]{\stackrel{#1}{\longrightarrow}} % transition
\newcommand{\dto}[1]{\mathrel{\stackrel{#1\ }         % abstract
      {\raisebox{0pt}[4pt][0pt]{$\Longrightarrow$}}}} % transition
\newcommand{\plat}[1]{\raisebox{0pt}[0pt][0pt]{#1}}   % no vertical space
\newcommand{\weg}[1]{}                                % omitted material

\newcommand{\Proc}{\mathbbm{P}}
\newcommand{\Tests}{\mathbbm{T}}
\newcommand{\Outcomes}{\mathbbm{O}}
\newcommand{\Apply}{\mathcal A\it pply}
\newcommand{\infr}[2]                                 % inference rule
{\rule{0mm}{6mm} \begin{array}{c} #1\\[0.1ex] \hline \rule{0ex}{2.7ex}#2 \end{array}}
\newcommand{\RL}{L}                                   % restricted set of actions
\newcommand{\mylabel}[1]{\hypertarget{lab:#1}{\ \mbox{{\scriptsize\sc (#1)}}}}
\newcommand{\myref}[1]{\hyperlink{lab:#1}{\scriptsize\sc (#1)}}
\newcommand{\rec}[2][X]{{\bf fix}\llparenthesis#1{:}#2\rrparenthesis} % recursion construct
\newcommand{\defis}{\stackrel{{\it def}}{=}}          % recursive definition
\newcommand{\dom}{{\it dom}}                          % domain
\newcommand{\dcup}{\stackrel{\mbox{\huge .}}{\cup}}   % disjoint union
\newcommand{\Comp}{\textit{Comp}}                     % set of computations
\newcommand{\ptr}{{\it ptraces}}                      % partial traces

\newcommand{\deadlocks}{{\it deadlocks}}
\newcommand{\ct}{{\it CT}}                            % complete traces
\newcommand{\infinite}{{\it inf}}                     % infinite traces
\newcommand{\infd}{{\it inf\!\!}_\bot}
\newcommand{\infdd}{{\it inf\!\!}_d}
\newcommand{\ini}{{\it initials}}
\newcommand{\failures}{{\it failures}}
\newcommand{\faild}{{\it failures\!}_\bot}
\newcommand{\faildd}{{\it failures\!}_d}
\newcommand{\diverg}{{\it divergences}}
\newcommand{\divd}{{\it divergences\!}_\bot}

\newcommand{\Hleq}{\sqsubseteq_{\rm Ho}}
\newcommand{\Sleq}{\sqsubseteq_{\rm Sm}}
\newcommand{\Mustleq}{\sqsubseteq_{\rm{must}}}
\newcommand{\Mayleq}{\sqsubseteq_{\rm{may}}}
\newcommand{\Musteq}{\equiv_{\rm{must}}}
\newcommand{\Mayeq}{\equiv_{\rm{may}}}
\begin{document}

\def\titlerunning{Reward Testing Equivalences for Processes}
\def\authorrunning{Rob van Glabbeek}
\title{\titlerunning}
\author{\authorrunning
\institute{Data61, CSIRO, Sydney, Australia}
\institute{School of Computer Science and Engineering,
University of New South Wales, Sydney, Australia}
\email{rvg@cs.stanford.edu}
}
\maketitle

\begin{abstract}
May and must testing were introduced by De Nicola and Hennessy to define semantic
equivalences on processes. May-testing equivalence exactly captures safety properties,
and must-testing equivalence liveness properties. This paper proposes \emph{reward testing}
and shows that the resulting semantic equivalence also captures conditional liveness properties.
It is strictly finer than both the may- and must-testing equivalence.
\vspace{3pt}
\end{abstract}

\begin{abstract}
\emph{This paper is dedicated to Rocco De Nicola, on the occasion of his 65$^{\it th}$ birthday.
Rocco's work has been a source of inspiration to my own.}
\end{abstract}

\section*{Introduction}

The idea behind semantic equivalences $\equiv$ and refinement preorders $\sqsubseteq$ on processes
is that $P\equiv Q$ says, essentially, that for practical purposes processes $P$ and $Q$ are equally
suitable, i.e.\ one can be replaced for by the other without untoward side effects. Likewise,
$P \sqsubseteq Q$ says that for all practical purposes under consideration, $Q$ is at least as suitable
as $P$, i.e.\ it will never harm to replace $P$ by $Q$.
% Thus, $p \equiv q$ iff both $p \sqsubseteq q$ and $q \sqsubseteq p$.
To this end, $Q$ must have all relevant good properties that $P$ enjoys.
Among the properties that ought to be so \emph{preserved}, are \emph{safety properties}, saying that nothing bad will even happen,
and \emph{liveness properties}, saying that something good will happen eventually.

In the setting of the process algebra CCS, refinement preorders $\Mayleq$ and $\Mustleq$, and
associated semantic equivalences $\Mayeq$ and $\Musteq$, were proposed by De Nicola \& Hennessy in \cite{DH84}.
In \cite{vG10} I argue that $\Mayeq$ and $\Musteq$ are the coarsest equivalences that enjoy
some basic compositionality requirements\footnote{Namely being congruences for injective renaming and
  partially synchronous interleaving operators, or equivalently all operators of CSP, or
  equivalently the CCS operators parallel composition, restriction and relabelling.}
and preserve safety and liveness properties, respectively. Yet neither preserves so-called \emph{conditional liveness properties}.
\begin{figure}[h]
\vspace{-1ex}
\input{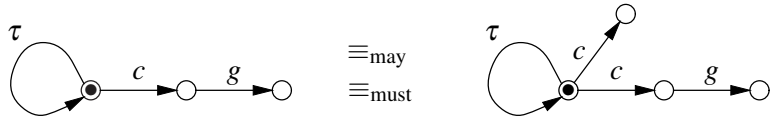}
\centerline{\raisebox{0em}{\box\graph}}
\caption{\it Processes identified  by may and must testing, but with different
  conditional liveness properties}
\label{conditional}
\end{figure}
This is illustrated in \fig{conditional}, showing two processes that are identified under both may and must testing.
From a practical point of view, the difference between
these two processes may be enormous. It could be that the action $c$
comes with a huge cost, that is only worth making when the good action $g$
happens afterwards. Only the right-hand side process is able to incur
the cost without any benefits, and for this reason it lacks an
important property that the left-hand process has. I call such
properties \emph{conditional liveness properties}.
A conditional liveness property says that
\begin{quote}\it
under certain conditions something good will eventually happen.
\end{quote}
This paper introduces a stronger form of testing that preserves conditional liveness properties.

\section{General setting}

It is natural to view the semantics of processes as being determined
by their ability to pass tests \cite{DH84,henn};
processes $P_1$ and $P_2$ are deemed to be semantically equivalent
unless there is a test which can distinguish them.  The actual tests
used typically represent the ways in which users, or indeed other
processes, can interact with $P_i$.
This idea can be formulated in the following general testing scenario \cite{DGHMZ07},
of which the testing scenarios of \cite{DH84,henn} are instances.  It assumes
\begin{itemize} 
\item a set of processes $\Proc$,\vspace{-2pt}
\item a set of tests $\Tests$, which can be applied to processes,\vspace{-2pt}
\item a set of outcomes $\Outcomes$, the possible results from applying a test to a process, and\vspace{-2pt}
\item a function $\Apply: \Tests \times \Proc \rightarrow \Pow^+(\Outcomes)$, representing
the possible results of applying a specific test to a specific process.
\end{itemize}
Here $\Pow^+(\Outcomes)$ denotes the collection of non-empty
subsets of $\Outcomes$; so the result of applying a test $T$ to a
process $P$, $\Apply(T,P)$, is in general a \emph{set} of outcomes,
representing the fact that the behaviour of processes, and indeed
tests, may be nondeterministic.

Moreover, some outcomes are considered better then others; for example
the application of a test may simply succeed, or it may fail, with
success being better than failure.  So one can assume that $\Outcomes$ is
endowed with a partial order, in which $o_1 \leq o_2$ means that $o_2$
is a better outcome than $o_1$.

When comparing the result of applying tests to processes one needs to
compare subsets of $\Outcomes$. There are two standard approaches to
make this comparison, based on viewing these sets as elements of
either the Hoare or Smyth powerdomain \cite{Hen82,AJ94} of $\Outcomes$.
For $O_1, O_2 \in \Pow^+(\Outcomes)$ let
\begin{enumerate}[(i)]
\item $O_1 \Hleq O_2$ if for every $o_1 \in O_1$ there exists some $o_2 \in O_2$
such that $o_1 \leq o_2$

\item $O_1 \Sleq O_2$ if for every $o_2 \in O_2$ there exists some $o_1 \in O_1$
such that $o_1 \leq o_2$.
\end{enumerate}
Using these two comparison methods one obtains two different semantic
preorders for processes:
\begin{enumerate}[(i)]
\item 
For $P,Q \in \Proc$ let $P \mathbin{\Mayleq} Q$ if
$\Apply(T,P) \mathbin{\Hleq} \Apply(T,Q)$ for every test $T$

\item 
Similarly, let  $P \Mustleq Q$ if
$\Apply(T,P) \Sleq \Apply(T,Q)$  for every test $T$.
\end{enumerate}  
Note that $\Mayleq$ and $\Mustleq$ are reflexive and transitive, and hence preorders.
I use $P \Mayeq Q$ and $P \Musteq Q$ to denote the associated equivalences.

{\newcommand{\PupCl}[1]{\mbox{P}^\uparrow\kern -.1em (#1)}
The terminology \emph{may} and \emph{must} refers to the following
reformulation of the same idea.  Let ${\it Pass} \subseteq \Outcomes$
be an upwards-closed subset of $\Outcomes$, i.e.\ satisfying
$o'\geq o\in{\it Pass} \Rightarrow o'\in {\it Pass}$, thought of as
the set of outcomes that can be regarded as \emph{passing} a
test. Then one says that a process $P$ \emph{may} pass a test $T$ with
an outcome in {\it Pass}, notation ``$P$ {\bf may} ${\it Pass}~T$'',
if there is an outcome $o\in\Apply(P,T)$ with $o\in\textit{Pass}$,
and likewise $P$ \emph{must} pass a test $T$ with an outcome in {\it
Pass}, notation \mbox{``$P$ {\bf must} ${\it Pass}~T$''}, if for all
$o\in\Apply(P,T)$ one has $o\in\textit{Pass}$.  Now
\begin{eqnarray*}
P \Mayleq Q \mbox{~~iff~~} \forall T\mathbin\in\Tests\; \forall \textit{Pass}
 \in \PupCl{\Outcomes}\, (P~\textbf{may}~\textit{Pass}~T
 ~\Rightarrow~ Q~\textbf{may}~\textit{Pass}~T)~~\\
P \Mustleq Q \mbox{~~iff~~} \forall T\mathbin\in\Tests\; \forall \textit{Pass}
 \in \PupCl{\Outcomes}\, (P~\textbf{must}~\textit{Pass}~T
 ~\Rightarrow~ Q~\textbf{must}~\textit{Pass}~T)
\end{eqnarray*}
where $\PupCl{\Outcomes}$ is the set of upwards-closed
subsets of $\Outcomes$.
}

The original theory of testing \cite{DH84,henn} is obtained by using
as the set of outcomes $\Outcomes$ the two-point lattice\vspace{-1ex}
\begin{center}
  \begin{picture}(40,50)(0,0)
  \linethickness{0.5pt} 
  \put(20,0){$\bot$}
  \put(20,45){$\top$}
  \put(24.2,13.7){\line(0,1){25}}
  \end{picture}
\end{center}
with $\top$ representing the success of a test application, and $\bot$ failure. 

\section{CCS: The Calculus of Communicating Systems}\label{sec:CCS}

{\renewcommand{\eta}{\alpha}
\renewcommand{\ell}{\alpha}

\begin{table*}[t]
\caption{Structural operational semantics of CCS}
\label{tab:CCS}
\normalsize
\begin{center}
\framebox{$\begin{array}{c@{}c@{\qquad}c}
\alpha.E \goto{\alpha} E  \mylabel{Act} &
\displaystyle\frac{E_j \goto{\alpha} E_j'}{\sum_{i\in I}E_i \goto{\alpha} E_j'}\makebox[2em][l]{~~\small($j\in I$)  \mylabel{Sum}}\\[4ex]
\displaystyle\frac{E\goto{\eta} E'}{E|F \goto{\eta} E'|F} \mylabel{Par-l}&
\displaystyle\frac{E\goto{a} E' ,~ F \goto{\bar{a}} F'}{E|F \goto{\tau} E'| F'} \mylabel{Comm}&
\displaystyle\frac{F \goto{\eta} F'}{E|F \goto{\eta} E|F'} \mylabel{Par-r}\\[4ex]
\displaystyle\frac{E \goto{\ell} E'}{E\backslash \RL \goto{\ell}E'\backslash \RL}~~(\ell,\bar{\ell}\not\in\RL) \mylabel{Res}&
\displaystyle\frac{E \goto{\ell} E'}{E[f] \goto{f(\ell)} E'[f]} \mylabel{Rel} &
\displaystyle\frac{\rec[S_X]{S} \goto{\alpha} E}{\rec{S}\goto{\alpha}E} \mylabel{Rec}
\end{array}$}
\end{center}
\end{table*}

\noindent
CCS \cite{ccs} is parametrised with a set $\Ch$ of \emph{names};
\plat{$Act := \Ch\dcup\bar\Ch\dcup \{\tau\}$} is the set of {\em actions}, where $\tau$ is a special \emph{internal action}
and $\bar{\Ch} := \{ \bar{c} \mid c \in \Ch\}$ is the set of \emph{co-names}.
Complementation is extended to $\bar\Ch$ by setting $\bar{\bar{\mbox{$c$}}}=c$.
Below, $a$ ranges over $\A:=\Ch\cup\bar\Ch$ and $\alpha$ over $Act$.
A \emph{relabelling} is a function $f\!:\Ch\mathbin\rightarrow \Ch$; it extends to $Act$ by
$f(\bar{c})\mathbin=\overline{f(c)}$ and $f(\tau):=\tau$.
Let $\V$ be a set $X$, $Y$, \ldots of \emph{process variables}.
The set $\IE_{\rm CCS}$ of CCS expressions is the smallest set including:
\begin{center}
\begin{tabular}{@{}lll@{}}
$\alpha.E$  & for $\alpha\mathbin\in Act$ and $E\mathbin\in\IE_{\rm CCS}$ & \emph{action prefixing}\\
$\sum_{i\in I}E_i$ & for $I$ an index set and $E_i\mathbin\in\IE_{\rm CCS}$ & \emph{choice} \\
$E|F$ & for $E,F\mathbin\in\IE_{\rm CCS}$ & \emph{parallel composition}\\
$E\backslash \RL$ & for $\RL\subseteq\Ch$ and $E\mathbin\in\IE_{\rm CCS}$ & \emph{restriction} \\
$E[f]$ & for $f$ a relabelling and $E\mathbin\in\IE_{\rm CCS}$ & \emph{relabelling} \\
$X$ & for $X\mathbin\in\V$ & \emph{process variable}\\
$\rec{S}$ & for $S\!:\V\mathord\rightharpoonup \IE_{\rm CCS}$ and $X\in \dom(S)$ & \emph{recursion}.
\end{tabular}
\end{center}
The expression $\sum_{i\in \{1,2\}}\!\alpha_i.E_i$ is often written as $\alpha_1.E_1 {+} \alpha_2.E_2$,
$\sum_{i\in \{1\}}\!\alpha_i.E_i$ as $\alpha_1.E_1$, and $\sum_{i\in \emptyset}\!\alpha_i.E_i$ as ${\bf 0}$.\linebreak[3]
Moreover, one abbreviates $\alpha.{\bf 0}$ by $\alpha$, and $P\backslash\{c\}$ by $P\backslash c$.\vspace{2pt}
A partial function $S\!:\V\mathbin\rightharpoonup \IE_{\rm CCS}$ is called a
\emph{recursive specification},
and traditionally written as \plat{$\{Y\defis S(Y)\mid Y\mathbin\in \dom(S)\}$}.
A CCS expression $E$ is \emph{closed} if each occurrence of a process variable $Y$ in $E$ lays within a
subexpression $\rec{S}$ of $E$ with $Y\mathbin\in\dom(S)$; $\Proc_{\rm CCS}$, ranged over by
$P,Q,\dots$, denotes the set of closed CCS expressions, or \emph{processes}.

The semantics of CCS is given by the labelled transition relation
$\mathord\rightarrow \subseteq \Proc_{\rm CCS}\times Act \times\Proc_{\rm CCS}$, where transitions 
\plat{$P\ar{\ell}Q$} are derived from the rules of \autoref{tab:CCS}.
Here $\rec[S_X]{S}$ denotes the expression $S(X)$ (written $S_X$) with $\rec[Y]{S}$ substituted for
each free occurrence of $Y\!$, for all $Y\in\dom(S)$, while renaming bound variables in $S_X$ as
necessary to avoid name-clashes.

The process $\alpha.P$ performs the action $\alpha$ first and subsequently acts as $P$.
The choice operator $\sum_{i\in I}P_i$ may act as any of its arguments $P_i$, depending on which of these processes is able to act at all.
The parallel composition $P|Q$ executes an action from $P$, an action from $Q$, or in the case where
$P$ and $Q$ can perform complementary actions $a$ and $\bar{a}$, the process can perform a synchronisation, resulting in an internal action $\tau$.
The restriction operator $P \backslash \RL$
inhibits execution of the actions from $\RL$ and their complements. 
The relabelling $P[f]$ acts like process $P$ with all labels $\ell$ replaced by $f(\ell)$.
Finally, the rule for recursion says that a recursively defined process $\rec{S}$ behaves exactly as
the body $S_X$ of the recursive equation \plat{$X\defis S_X$}, but 
with recursive calls $\rec[Y]{S}$ substituted for the variables $Y\in\dom(S)$.

\section{Classical may and must testing for CCS}\label{sec:classical}

Let $Act^\omega:= Act \cup\{\omega\}$, where $\omega \notin Act$ is a special action reporting success.
A CCS \emph{test} $T\in\IT_{\rm CCS}$ is defined just like a CCS process, but with $\alpha$ ranging over $Act^\omega$.
So a CCS process is a special kind of CCS test, namely one that never performs the action $\omega$.
To apply the test $T$ to the process $P$ one runs them in parallel; that is, one runs the combined
process $T | P$---which is itself a CCS test.

\begin{definition}{computation}
A \emph{computation} $\pi$ is a finite or infinite sequence $T_0,T_1,T_2,\dots$ of tests, such that
(i) if $T_n$ is the final element in the sequence, then \plat{$T_n \goto\tau T$} for no $T$, and
(ii) otherwise \plat{$T_n\goto{\tau} T_{n+1}$}.

A computation $\pi$ is \emph{successful} if it contains a state $T$ with $T \goto{\omega} T'$ for some $T'$.

For $T\in\Tests_{\rm CCS}$, $P\in\Proc_{\rm CCS}$, let $\Comp(T,P)$ be the set of computations whose initial
element is $T|P$.

Let $\Apply(T,P):=\{\top \mid \exists \mbox{ successful $\pi\in \Comp(T,P)$}\}
    \cup          \{\bot \mid \exists \mbox{ unsuccessful $\pi\in \Comp(T,P)$}\}$.
\end{definition}
Using this definition of $\Apply$ it follows that $P \Mayleq Q$ holds unless there is a test $T$
such that $T|P$ has (that is, is the initial state of) a successful computation but $Q$ has not.
Likewise $P \Mustleq Q$  holds unless there is a test $T$
such that $T|P$ has only successful computations but $Q$ has not.

\section{Dual may and must testing}\label{sec:dual}

A \emph{liveness property} \cite{Lam77} is a property that says that \emph{something good will eventually happen}.
In the context of CCS, any test $T$ can be regarded to specify a liveness property; a process $P$
is defined to have this property iff all computations of $T|P$ are successful.
Now $P \Mustleq Q$ holds iff all liveness properties $T\in\Tests_{\rm CCS}$ that are enjoyed by $P$
also hold for $Q$.

A \emph{safety property} \cite{Lam77} is a property that says that \emph{something bad will never happen}.
When thinking of the special action $\omega$ as reporting that something bad has occurred, rather
than something good, any test $T$ can also be regarded to specify a safety property; a process $P$
is defined to have this property iff none of the computations of $T|P$ are catastrophic; here
\emph{catastrophic} is simply another word for ``successful'', when reversing the connotation of $\omega$.
Now $Q \Mayleq P$ holds iff all safety properties $T\in\Tests_{\rm CCS}$ that are enjoyed by $P$
also hold for $Q$.

A \emph{labelled transition system} (LTS) over a set $Act$ is a pair $(\Proc,\rightarrow)$ where $\Proc$ is a set of
\emph{processes} or \emph{states} and ${\rightarrow}\subseteq \Proc \times Act \times \Proc$ a set of \emph{transitions}.
In \cite{vG10} preorders $\sqsubseteq_{\it liveness}$ and $\sqsubseteq_{\it safety}$ are defined on
LTSs. Specialised to the LTS $(\Proc_{\rm CCS},\rightarrow)$ induced by CCS,
$\sqsubseteq_{\it liveness}$ coincides with $\Mustleq$, and $\sqsubseteq_{\it safety}$ is exactly the
reverse of $\Mayleq$, in accordance with the reasoning above.

To explain the reversal of $\Mayleq$ when dealing with safety properties,
I propose a variant of CCS testing where in \df{computation} the word ``catastrophic'' is used
for ``successful'' and $\Apply$ is redefined by\vspace{-1pt}
\[\Apply(T,P):=\{\bot \mid \exists \mbox{ catastrophic $\pi\in \Comp(T,P)$}\}
 \cup          \{\top \mid \exists \mbox{ uncatastrophic $\pi\in \Comp(T,P)$}\}.\vspace{-1pt}\]
An equivalent alternative to redefining $\Apply$ is to simply invert the order between $\bot$ and $\top$.
Let \plat{$\Mayleq^{\rm dual}$} and  $\Mustleq^{\rm dual}$ be the versions of the may- and must-testing
preorders obtained from this alternative definition.
It follows immediately from the definitions that $P \Mayleq^{\rm dual} Q$ iff $Q \Mustleq P$
and that $P \Mustleq^{\rm dual} Q$ iff $Q \Mayleq P$.
Based on this, it may be more accurate to say that  $\sqsubseteq_{\it safety}$ coincides with \plat{$\Mustleq^{\rm dual}$}.

A \emph{possibility property} \cite{Lam98}  is a property that says that \emph{something good might eventually happen}.
A test $T$ can be regarded to specify a possibility property; a process $P$
is defined to have this property iff some computation of $T|P$ is successful.
Now $P \Mayleq Q$ holds iff all possibility properties $T\in\Tests_{\rm CCS}$ that are enjoyed by $P$
also hold for $Q$. Lamport argues that ``verifying possibility
properties tells you nothing interesting about a system'' \cite{Lam98}.
As an example, consider the following models of coffee machines:\vspace{-1pt}
$$C_1:= \tau   \qquad\qquad  C_2:= \tau.c+\tau   \qquad\qquad C_3:=\tau.c\pagebreak[2]$$
where $c$ is the act of dispensing coffee. The machine $C_1$ surely will not make coffee, $C_2$
makes a nondeterministic choice between making coffee or not, and $C_3$ surely makes coffee.
Under may testing, systems $C2$ and $C_3$ are equivalent---both have the possibility of making
coffee---and each of them is better than $C_1$: $C_1 \sqsubset_{\rm may} C_2 \Mayeq C_3$.
The relevance of this indeed is questionable. It takes must testing to formalise that $C_3$ is
better than $C_2$: only $C_3$ guarantees that coffee will eventually be dispensed.

When employing dual testing, the same example applies, but with $c$ denoting a catastrophe.
Now $C_1$ is safe, whereas $C_2$ and $C_3$ are not:  $C_1 \sqsupset^{\rm dual}_{\rm must} C_2 \equiv^{\rm dual}_{\rm must} C_3$.
Dual may testing would argue that $C_2$ is better than $C_3$ because a catastrophe might be avoided.
This however, can be deemed a weak argument.

In view of these considerations, I will focus on the preorders $\Mustleq$ and $\Mustleq^{\rm dual}$
(or $\sqsubseteq_{\it safety}$).
The (dual) may preorders simply arise as their inverses, and hence do not require explicit treatment.

\section{Reward testing for CCS}

A CCS \emph{reward test} is defined just like a CCS process, but with $\alpha$ ranging over
$Act\times\IR$, the \emph{valued actions}.
A valued action is an action tagged with a real number, the \emph{reward} for executing this action.
A negative reward can be seen as a penalty.
Let $\Tests^R_{\rm CCS}$ be the set of CCS reward tests.
The structural operational semantics for CCS reward tests has the following modified rules:
\begin{center}
\framebox{$\begin{array}{c@{\qquad}c@{\qquad}c}
\displaystyle\frac{P\ar{a,r} P' ,~ Q \ar{\bar{a},r'} Q'}{P|Q \ar{\tau,r+r'} P'| Q'} \mylabel{Comm$'$}&
\displaystyle\frac{P \ar{\ell,r} P'}{P\backslash \RL \ar{\ell,r}P'\backslash \RL}~~(\ell,\bar{\ell}\not\in\RL) \mylabel{Res$'$}&
\displaystyle\frac{P \ar{\ell,r} P'}{P[f] \ar{f(\ell),r} P'[f]} \mylabel{Rel$'$}
\end{array}$}
\end{center}
Thus, in synchronising two actions one reaps the rewards of both.
In all other rules of \tab{CCS}, $\alpha$ is simply replaced by $\alpha,r$, with $r\in\IR$.
A valued action $\alpha,0$ is simply denoted $\alpha$,
so that a CCS process can be seen as a special CCS reward test, namely one in which all rewards are $0$.
To apply a reward test $T$ to a process $P$ one again runs them in parallel.
}

\begin{definition}{reward computation}
A \emph{reward computation} $\pi$ is a finite or infinite sequence $T_0,r_1,T_1,r_2,T_2\dots$ of reward tests, such that
(i) if $T_n$ is the final element in $\pi$, then \plat{$T_n \ar{\tau,r} T$} for no $r$ and $T$, and
(ii) otherwise \plat{$T_n\ar{\tau,r_{n+1}} T_{n+1}$}.

The \emph{reward} of a finite computation $\pi$ ending in $T_n$ is $\sum_{i=1}^n r_i$.
The \emph{reward} of an infinite computation  $T_0,r_1,T_1,r_2,T_2\dots$
is\vspace{-2ex} $$\displaystyle\inf_{n\rightarrow\infty}\sum_{i=1}^n r_i \qquad \in \IR\cup\{-\infty,\infty\}.$$
For $T\in\Tests^R_{\rm CCS}$, $P\in\Proc_{\rm CCS}$, let $\Comp^R(T,P)$ be the set of reward computations with initial
element $T|P$.
\\
Let $\Apply(T,P):=\{\textit{reward}(\pi) \mid \pi\in \Comp^R(T,P)\}$.
\end{definition}
This defines reward preorders $\sqsubseteq^{\rm may}_{\rm reward}$ and $\sqsubseteq^{\rm must}_{\rm reward}$ on $\Proc_{\rm CCS}$.
It will turn out that $P \sqsubseteq^{\rm may}_{\rm reward} Q$ iff $Q \sqsubseteq^{\rm must}_{\rm reward} P$.
As a consequence I will focus on \plat{$\sqsubseteq^{\rm must}_{\rm reward}$}, and simply call it $\sqsubseteq_{\rm reward}$.

\section{Characterising reward testing}\label{sec:explicit}

Assuming a fixed LTS $(\mathbb{P},\rightarrow)$, labelled over a set $Act = \A \dcup \{\tau\}$,
the ternary relation $\mathord{\dto\_} \subseteq \mathbb{P} \times \A^*
\times \mathbb{P}$ is the least relation satisfying\\[2pt]
\mbox{}
\hfill
  $P \dto \epsilon P$\enskip,
\hfill
 $\infr{P \goto \tau Q}{P \dto \epsilon Q}$\enskip,
\hfill
 $\infr{P \goto a Q,~ a \not= \tau}{P \dto {a} Q}$
\hfill and \hfill
 $\infr{P \dto \sigma Q \dto \rho r}{P \dto {\sigma\rho} r}$
\enskip.
\hfill\mbox{}\\[2pt]
For $\sigma\in\A^*$ and $\nu\in\A^*\cup\A^\infty$
write $\sigma\leq\nu$ for ``$\sigma$ is a prefix of~$\rho$'', i.e.\
``$\exists \rho\;.\sigma\rho=\nu$''.
\pagebreak[3]

\begin{definition}{traces}
Let $P\in \mathbb{P}$.
\begin{itemise}\vspace{-1ex}
\item $a_1 a_2 a_3 \cdots \in \A^\infty$ is an \emph{infinite trace}
  of $P$ if there are $P_1,P_2,\ldots$ such that
  $P\dto{a_1}P_1\dto{a_2}P_2\dto{a_3} \cdots$.
\item $\infinite(P)$ denotes the set of infinite traces of $P$.
\item $P$ \emph{diverges}, notation $P{\Uparrow}$, if there are
  $P_i\in\mathbb{P}$ for all $i>0$ such that
  \plat{$P\goto{\tau}P_1\goto{\tau}P_2\goto{\tau} \cdots$}.
\item \plat{$\diverg(P):= \{\sigma \in \A^* \mid \exists Q.\; P \dto \sigma Q {\Uparrow}\}$}
      is the set of \emph{divergence traces} of $P$.
\item \plat{$\ini(P):=\{\alpha\in \A \mid \exists Q.\; P \goto{\alpha} Q\}$}.
\item \plat{$\deadlocks(P):= \{\sigma \in \A^* \mid \exists Q.\; P \dto\sigma Q \wedge \ini(Q)=\emptyset\}$}
      is the set of \emph{deadlock traces} of $P$.
\item $\ct(P):=\infinite(P)\cup\diverg(P)\cup\deadlocks(P)$
      is the set of \emph{complete traces} of $P$.
\item \plat{$\ptr(P) := \{\sigma \in \A^* \mid \exists Q.\; P \dto \sigma Q\}$}
      is the set of \emph{partial traces} of $P$.
% \item $\traces(P):=\infinite(P)\cup\ptr(P)$ is the set of \emph{traces} of $P$.
\item \plat{$\failures(P):=\{\langle \sigma, X\rangle \in \A^*\times\Pow(\A) \mid
  \exists Q.\; P \dto{\sigma} Q \wedge \ini(Q) \cap (X\cup\{\tau\})=\emptyset\}$}.
\item $\faildd(P):= \failures(P) \cup \{\langle\sigma,X\rangle \mid
  \sigma\in\diverg(P)\wedge X\subseteq \A\}$.
\item $\infdd(P):= \infinite(P) \cup \{\nu \in\A^\infty \mid
  \forall\sigma{<}\nu\; \exists \rho\in\diverg(P).\; \sigma\leq\rho<\nu\}$.
\item $\divd(P):= \{\sigma\rho \mid \sigma\in\diverg(P) \wedge \rho\in \A^*\}$.
\item $\infd(P):= \infinite(P) \cup \{\sigma\nu \mid \sigma\in\diverg(P) \wedge \nu\in \A^\infty\}$.
\item $\faild(P):= \failures(P) \cup \{\langle\sigma\rho,X\rangle \mid
  \sigma\in\diverg(P) \wedge \rho\in \A^* \wedge X\subseteq \A\}$.
\end{itemise}
\end{definition}
% Note that $\traces(P) = \{\sigma \in \A^\omega \mid \exists \rho\in\ct(P).\; \sigma \leq \rho\}$.
Note that \hfill $\ptr(R)=\{\sigma\mid \langle\sigma,\emptyset\rangle \in \faildd(R)\}$ \hfill for any $R\in \Proc$. \hfill \hypertarget{star}{(*)}
\\[1ex]
A \emph{path} of a process $P\in\Proc$ is an alternating sequence
$P_0\,\alpha_1\,P_1\,\alpha_2\,P_2\cdots$ of processes/states and actions, starting with a state and
either being infinite or ending with a state, such that $P_i \ar{\alpha_{i+1}} P_{i+1}$ for all relevant~$i$.
Let $l(\pi):= \alpha_1 \alpha_2 \cdots$ be the sequence of actions in $\pi$, and $\ell(\pi)$ the
same sequence after all $\tau$s are removed.
Now $\sigma \in \infinite(P) \cup \diverg(P)$ iff $P$ has an infinite path $\pi$ with $\ell(\pi)=\sigma$.
Likewise, $\sigma \in \ptr(P)$ iff $P$ has a finite path $\pi$ with $\ell(\pi)=\sigma$.
Finally, $\sigma \in \infinite(P) \cup \ptr(P)$ iff $P$ has an path $\pi$ with $\ell(\pi)=\sigma$.

\newcommand{\startingfrom}{of\ }%
Any transition \plat{$P|Q \ar{\alpha} R$} derives, through the
rules of \tab{CCS}, from
\begin{itemise}
\vspace{-1ex}
\item a transition \plat{$P \ar{\alpha} P'$} and a state $Q$, where $R=P'|Q$\,,
\item two transitions \plat{$P \ar{a_1} P'$ and $Q \ar{\bar a_2} Q'$}, where $R=P'|Q'$\,,
\item or from a state $P$ and a transition \plat{$Q \ar{\alpha} Q'$}, where $R=P|Q'$.
\vspace{-1ex}
\end{itemise}
This transition/state, transition/transition or state/transition pair is called a 
\emph{decomposition} of \plat{$P|Q \ar{\alpha} R$}; it need not be unique.
Now a \emph{decomposition} of a path $\pi$ \startingfrom $P|Q$ into paths $\pi_1$ and $\pi_2$
\startingfrom $P$ and $Q$, respectively, is obtained by decomposing each transition in the path, and
concatenating all left-projections into a path \startingfrom $P$ and all right-projections into a
path \startingfrom $Q$---notation $\pi \in \pi_1 | \pi_2$ \cite{GH15a}.
Here it could be that $\pi$ is infinite, yet either $\pi_1$ or $\pi_2$ (but not both) are finite.
Again, decomposition of paths need not be unique.

\begin{theorem}{reward characterisation}
Let $P,Q\in\Proc_{\rm CCS}$. Then
$P \sqsubseteq_{\rm reward} Q ~~\Leftrightarrow~~
\begin{array}[t]{@{}r@{~\supseteq~}l@{}}
\diverg(P) & \diverg(Q) \wedge \mbox{}\\
\infinite(P) & \infinite(Q) \wedge \mbox{}\\
\faildd(P) & \faildd(Q).
\end{array}$
\end{theorem}

\begin{proof}
Let $\sqsubseteq_{\it NDFD}$ be the preorder defined by:
$P\sqsubseteq_{\it NDFD} Q$ iff the right-hand side of \thm{reward characterisation} holds.

For $\sigma = a_1 a_2 \cdots a_n \in \A^*$, let $\bar\sigma.T$ with $T\in\Tests^R_{\rm CCS}$ be the CCS reward test
$\bar a_1 . \bar a_2 . \cdots \bar a_1 . T$. It starts with performing the complements of the actions
in $\sigma$, where each of these actions is given a reward $0$. 

Write $\alpha^r$ for $(\alpha,r)\in Act\times\IR$.
For $\nu = a_1 a_2 a_3 \cdots \in \A^\infty$, let $\bar \nu^r$ be the CCS reward test $\rec[X_0]{S}$
where \plat{$S=\{X_i \defis \bar a_{i+1}^r. X_{i+1}\mid i\geq 0\}$}. This test simply performs the infinite sequence
of complements of the actions in $\nu$, where each of these actions is given a reward $r$.
\vspace{1ex}

\noindent
``$\Rightarrow$'':
Suppose $P \not\sqsubseteq_{\it NDFD} Q$.

Case 1: Let $\sigma \in \diverg(Q)\setminus \diverg(P)$.
Take $T:=\bar\sigma.\tau^{-1}.\tau^1 \in \Tests^R_{\rm CCS}$.
Then $T|Q$ has a computation $\pi$ with $\textit{reward}(\pi)<0$,
whereas $T|P$ has no such computation. Hence $P \not\sqsubseteq_{\rm reward} Q$.

Case 2: Let $\nu \in \infinite(Q)\setminus \infinite(P)$.
Take $T:=\bar\nu^{-1} \in  \Tests^R_{\rm CCS}$.
Then $T|Q$ has a computation $\pi$ with $\textit{reward}(\pi)=-\infty$,
whereas $T|P$ has no such computation. Hence $P \not\sqsubseteq_{\rm reward} Q$.

Case 3: Let $\langle \sigma,X\rangle \in \faildd(Q)\setminus \faildd(P)$.
Take $T:=\bar\sigma.\tau^{-1}.\sum_{a\in X} a^1 \in \Tests^R_{\rm CCS}$.
Then $T|Q$ has a computation $\pi$ with $\textit{reward}(\pi)<0$,
whereas $T|P$ has no such computation. Hence $P \not\sqsubseteq_{\rm reward} Q$.
\vspace{1ex}

\noindent
``$\Leftarrow$'': 
Suppose $P \sqsubseteq_{\it NDFD} Q$.
Let $T\in\Tests^R_{\rm CCS}$ and $r\in\IR$ be such that $\exists \pi\in\Comp(T|Q)$ with
$\textit{reward}(\pi)=r$.
It suffices to find a $\pi'\in\Comp(T|P)$ with $\textit{reward}(\pi')\leq r$.
The computation $\pi$ can be seen as a path of $T|Q$ in which all actions are $\tau$.
Decompose this path into paths $\pi_1$ of $T$ and $\pi_2$ of $Q$.
Note that $\textit{reward}(\pi)=\textit{reward}(\pi_1)$.

Case 1: Let $\pi_2$ be infinite. Then $\ell(\pi_2) \in \infinite(Q) \cup\diverg(Q) \subseteq \infinite(P) \cup\diverg(P)$.
Thus $P$ has an infinite path $\pi'_2$ with $\ell(\pi'_2)=\ell(\pi_2)$. Consequently, $T|P$ has an infinite path $\pi' \in \pi_1|\pi'_2$
that is a computation with $\textit{reward}(\pi')=r$.

Case 2: Let $\pi_2$ be finite and $\pi_1$ be infinite. Then $\ell(\pi_1)\in\diverg(T)$ and
$\ell(\pi_2) \in \ptr(Q) \subseteq \ptr(P)$. The latter inclusion follows by \hyperlink{star}{(*)}.
Thus $P$ has a finite path $\pi'_2$ with $\ell(\pi'_2)=\ell(\pi_2)$. Consequently, $T|P$ has an infinite path $\pi' \in \pi_1|\pi'_2$ that is a computation with $\textit{reward}(\pi')=r$.

Case 3: Let $\pi_1$ and $\pi_2$ be finite. Let $T'$ and $Q'$ be the last states of $\pi_1$ and
$\pi_2$, respectively. Let $X:=\{a\in Act \mid a^r\in\ini(T')\}$.
Then $\tau\notin X$, $\tau\notin \ini(Q')$ and $\ini(Q')\cap X=\emptyset$.
So $\langle \ell(\pi_2), X\rangle \in \failures(Q) \subseteq \faildd(Q) \subseteq \faildd(P)$.
Thus $P$ has either an infinite path $\pi'_2$ with $\ell(\pi'_2)=\ell(\pi_2)$ or
a finite path $\pi'_2$ with $\ell(\pi'_2)=\ell(\pi_2)$ and whose last state $P'$ satisfies $\ini(P')\cap (X\cup\{\tau\})=\emptyset$.
Consequently, $T|P$ has a finite or infinite path $\pi' \in \pi_1|\pi'_2$ that is a computation with $\textit{reward}(\pi')=r$.
\end{proof}

\section{Weaker notions of reward testing}\label{sec:weaker}
\newcommand{\fpr}{\textrm{\scriptsize fp-reward}}
\newcommand{\spr}{\textrm{\scriptsize sp-reward}}
\newcommand{\nnr}{+\textrm{\scriptsize reward}}
\newcommand{\npr}{-\textrm{\scriptsize reward}}
\newcommand{\fpnpr}{\textrm{\scriptsize fp-}-\textrm{\scriptsize reward}}
\newcommand{\spnpr}{\textrm{\scriptsize sp-}-\textrm{\scriptsize reward}}

\emph{Finite-penalty reward testing} doesn't allow computations that incur infinitely many penalties.
A test $T\in\Tests^R_{\rm CCS}$ has \emph{finite penalties} if each infinite path $T \alpha_1^{r_1} T_1 \alpha_2^{r_2} T_2 \cdots $
has only finitely many transitions $i$ with $r_i<0$.
Let  $P \sqsubseteq_{\rm \fpr} Q$ iff $\Apply(T,P) \Sleq \Apply(T,Q)$ for every finite-penalty reward test $T\!$.

\begin{theorem}{finite-penalty reward characterisation}
Let $P,Q\in\Proc_{\rm CCS}$. Then
$P \sqsubseteq_{\fpr} Q ~~\Leftrightarrow~~
\begin{array}[t]{@{}r@{~\supseteq~}l@{}}
\diverg(P) & \diverg(Q) \wedge \mbox{}\\
\infdd(P) & \infdd(Q) \wedge \mbox{}\\
\faildd(P) & \faildd(Q).
\end{array}$
\end{theorem}

\begin{proof}
Let $\sqsubseteq_{\it FDI}^d$ be the preorder defined by:
$P\sqsubseteq_{\it FDI}^d Q$ iff the right-hand side of \thm{finite-penalty reward characterisation} holds.
\vspace{1ex}

\noindent
``$\Rightarrow$'':
Suppose $P \not\sqsubseteq_{\it FDI}^d Q$.
Case 1 and 3 proceed exactly as in the proof of \thm{reward characterisation},
but the proof of Case 2 needs to be revised, as its proof uses a test with infinitely many penalties.
So assume $$\diverg(P) \supseteq \diverg(Q) \quad\wedge\quad \faildd(P) \supseteq \faildd(Q)$$ and let
$\nu \in \infdd(Q)\setminus \infdd(P)$.
I can rule out the case $\forall\sigma{<}\nu\; \exists \rho\in\diverg(q).\; \sigma\leq\rho<\nu$
because then $\nu \in \infdd(P)$, using that $\diverg(Q)\subseteq\diverg(P)$.
So $\nu \in \infinite(Q)$. Let $\nu:=\nu_1\nu_2$, where each $\rho\in\diverg(Q)$ with $\rho<\nu$
satisfies $\rho<\nu_1$. Let $\nu_2=b_1 b_2 \cdots \in\A^\infty$.\vspace{2pt}
Take $T:=\bar\nu_1.\tau^{-1}.\rec[Y_0]{S}$,
where \plat{$S=\{Y_i \defis \tau^1 + \bar b_{i+1}. Y_{i+1} \mid i\geq 0\}$}.
Then $T|Q$ has a computation $\pi$ with $\textit{reward}(\pi)<0$,
whereas $T|P$ has no such computation. Hence $P \not\sqsubseteq_{\fpr} Q$.
\vspace{1ex}

\noindent
``$\Leftarrow$'': 
Suppose $P \sqsubseteq_{\it FDI}^d Q$. The proof proceeds just as the one of \thm{reward characterisation},
except for Case 1.

Case 1: Let $\pi_2$ be infinite. Then $\ell(\pi_2) \in \infinite(Q) \cup\diverg(Q) \subseteq \infdd(P) \cup\diverg(P)$.
In case $\ell(\pi_2) \in \infinite(P) \cup\diverg(P)$ the proof concludes as for \thm{reward characterisation}.
So assume that $\ell(\pi_2)\in\A^\infty$ and $\forall\sigma{<}\ell(\pi_2)\; \exists \rho\in\diverg(P).\; \sigma\leq\rho<\ell(\pi_2)$.
Then there are prefixes $\pi^\dagger$,  $\pi_1^\dagger$ and $\pi_2^\dagger$ of $\pi$, $\pi_1$ and
$\pi_2$ such that (i) \plat{$\pi^\dagger\in\pi_1^\dagger|\pi_2^\dagger$}, (ii) there are no negative rewards
allocated in the suffix of $\pi_1$ past \plat{$\pi_1^\dagger$}, and (iii) \plat{$\ell(\pi_2^\dagger) \in\diverg(P)$}.
Let $\pi_2'$ be an infinite path of $P$ with $\ell(\pi_2')=\ell(\pi_2^\dagger)$.
Then there is a computation $\pi'\in\pi_1^\dagger|\pi_2'$ of $T|P$ with $\textit{reward}(\pi') =
\textit{reward}(\pi_1^\dagger) \leq \textit{reward}(\pi_1) =r$.
\end{proof}

\emph{Single penalty reward testing} doesn't allow computations that incur multiple penalties.
A test $T\in\Tests^R_{\rm CCS}$ has the \emph{single penalty} property if each path $T \alpha_1^{r_1} T_1 \alpha_2^{r_2} T_2 \cdots $
has at most one transition $i$ with $r_i<0$.
Let  $P \sqsubseteq_{\spr} Q$ iff $\Apply(T,P) \Sleq \Apply(T,Q)$ for every single penalty reward
test $T\!$.
Obviously, $\sqsubseteq_{\spr}$ coincides with $\sqsubseteq_{\fpr}$.
This follows because all test used in the proof of \thm{finite-penalty reward characterisation} have
the single penalty property.

Analogously one might weaken reward testing and/or single penalty reward testing by requiring that in
each computation only finitely many, or at most one, positive reward can be reaped. This does not constitute a real weakening,
as the tests used in Theorems~\ref{thm:reward characterisation} and~\ref{thm:finite-penalty reward characterisation}
already allot at most a single positive reward per computation only.

\emph{Nonnegative reward testing} requires all rewards to be nonnegative.
Let $P \sqsubseteq_{\nnr} Q$ iff $\Apply(T,P) \linebreak[1] \Sleq \Apply(T,Q)$ for every nonnegative reward test $T\!$.
Likewise $\sqsubseteq_{\npr}$ requires all rewards to be $0$ or negative.

\begin{theorem}{nonnegative reward characterisation}
Let $P,Q\in\Proc_{\rm CCS}$. Then
$P \sqsubseteq_{\nnr} Q ~~\Leftrightarrow~~
\begin{array}[t]{@{}r@{~\supseteq~}l@{}}
\divd(P) & \divd(Q) \wedge \mbox{}\\
\infd(P) & \infd(Q) \wedge \mbox{}\\
\faild(P) & \faild(Q).
\end{array}$
\end{theorem}

\begin{proof}
Let $\sqsubseteq_{\it FDI}^\bot$ be the preorder defined by:
$P\sqsubseteq_{\it FDI}^\bot Q$ iff the right-hand side of \thm{nonnegative reward characterisation} holds.
\vspace{1ex}

\noindent
``$\Rightarrow$'':
Suppose $P \not\sqsubseteq_{\it FDI}^\bot Q$.

Case 1: Let $\sigma =a_1 a_2 \cdots a_n \in \divd(Q)\setminus \divd(P)$.
Take $T:=\rec[X_0]{S}$ in which\vspace{-1ex}
$$S = \{X_i \defis \tau^1 + a_{i+1}.X_{i+1} \mid 0\mathop{\leq} i \mathop{<}  n\}\cup \{X_n \defis \tau^1\}.\vspace{-1ex}$$
Then $T|Q$ has a computation $\pi$ with $\textit{reward}(\pi)<1$,
which $T|P$ has not. Hence $P \not\sqsubseteq_{\nnr} Q$.

Case 2: Let $\nu =a_1 a_2 \cdots \in \infd(Q)\setminus \infd(P)$.
Let $T:=\rec[X_0]{S}$ with \plat{$S = \{X_{i-1} \defis \tau^1 \!+ a_{i}.X_{i} \mid i\mathop{\geq} 1\}$}.
Then $T|Q$ has a computation $\pi$ with $\textit{reward}(\pi)<1$,
which $T|P$ has not. Hence $P \not\sqsubseteq_{\nnr} Q$.

Case 3: Let $\langle a_1 a_2 \cdots a_n,X\rangle \in \faild(Q)\setminus \faild(P)$.
Take $T:=\rec[X_0]{S}$ in which\vspace{-1ex}
$$S = \{X_i \defis \tau^1 + a_{i+1}.X_{i+1} \mid  0\mathop{\leq} i \mathop{<}  n\}\cup \{X_n \defis \sum_{a\in X} a^1\}.\vspace{-1ex}$$
Then $T|Q$ has a computation $\pi$ with $\textit{reward}(\pi)<1$,
which $T|P$ has not. Hence $P \not\sqsubseteq_{\nnr} Q$.
\vspace{1ex}

\noindent
``$\Leftarrow$'': 
Suppose $P \sqsubseteq_{\it FDI}^\bot Q$.
Let $T\in\Tests^R_{\rm CCS}$ be a nonnegative rewards test and $r\in\IR$ be such that
there is a $\pi\in\Comp(T|Q)$ with $\textit{reward}(\pi)=r$.
It suffices to find a $\pi'\in\Comp(T|P)$ with $\textit{reward}(\pi')\leq r$.
The computation $\pi$ can be seen as a path of $T|Q$ in which all actions are $\tau$.
Decompose this path into paths $\pi_1$ of $T$ and $\pi_2$ of $Q$.
Note that $\textit{reward}(\pi)=\textit{reward}(\pi_1)$.

Case 1: Let $\pi_2$ be infinite. Then $\ell(\pi_2) \in \infinite(Q) \cup\diverg(Q) \subseteq \infd(P) \cup\divd(P)$.
If $\ell(\pi_2) \in \infinite(P) \cup\diverg(P)$ then $P$ has an infinite path $\pi'_2$ with
$\ell(\pi'_2)=\ell(\pi_2)$. Consequently, $T|P$ has an infinite path $\pi' \in \pi_1|\pi'_2$ that is
a computation with $\textit{reward}(\pi')=r$.
The alternative is that $\ell(\pi_2)$ has a prefix in $\diverg(P)$.
In that case there are prefixes $\pi^\dagger$,  $\pi_1^\dagger$ and \plat{$\pi_2^\dagger$} of $\pi$, $\pi_1$ and
$\pi_2$ such that \plat{$\pi^\dagger\in\pi_1^\dagger|\pi_2^\dagger$} and \plat{$\ell(\pi_2^\dagger) \in\diverg(P)$}.
Let $\pi_2'$ be an infinite path of $P$ with $\ell(\pi_2')=\ell(\pi_2^\dagger)$.
Then there is a computation $\pi'\in\pi_1^\dagger|\pi_2'$ of $T|P$ with
$\textit{reward}(\pi') = \textit{reward}(\pi_1^\dagger) \leq \textit{reward}(\pi_1) =r$.

Case 2: Let $\pi_2$ be finite and $\pi_1$ be infinite. Then $\ell(\pi_1)\in\diverg(T)$ and
$\ell(\pi_2) \in \ptr(Q) \subseteq \ptr(P)\cup\divd(P)$. The latter inclusion follows since\vspace{-1ex}
$$\ptr(R)\cup\divd(R)=\{\sigma\mid \langle\sigma,\emptyset\rangle \in \faild(R)\}\vspace{-1ex}$$ for any $R\in \Proc$.
If $\ell(\pi_2) \in \ptr(P)$ then
$P$ has a finite path $\pi'_2$ with $\ell(\pi'_2)=\ell(\pi_2)$. Consequently, $T|P$ has an
infinite path $\pi' \in \pi_1|\pi'_2$ that is a computation with $\textit{reward}(\pi')=r$.
The alternative is handled just as for Case 1 above.

Case 3: Let $\pi_1$ and $\pi_2$ be finite. Let $T'$ and $Q'$ be the last states of $\pi_1$ and
$\pi_2$, respectively. Let $X:=\{a\in Act \mid a^r\in\ini(T')\}$.
Then $\tau\notin X$, $\tau\notin \ini(Q')$ and $\ini(Q')\cap X=\emptyset$.
So $\langle \ell(\pi_2), X\rangle \in \failures(Q) \subseteq \faild(Q) \subseteq \faild(P)$.
If $\langle \ell(\pi_2)\in\failures(P)$ then $P$ has a finite path $\pi'_2$ with
$\ell(\pi'_2)=\ell(\pi_2)$ and whose last state $P'$ satisfies $\ini(P')\cap (X\cup\{\tau\})=\emptyset$. 
Consequently, $T|P$ has a finite or infinite path $\pi' \in \pi_1|\pi'_2$ that is a computation with $\textit{reward}(\pi')=r$.
The alternative is handled just as for Case 1 above.
\end{proof}
One might weaken nonnegative reward testing by requiring that in
each computation only finitely many, or at most one, reward can be reaped. This does not constitute a real weakening,
as the tests used in \thm{nonnegative reward characterisation} already allot at most a single reward per computation only.

\begin{theorem}{nonpositive reward characterisation}
Let $P,Q\in\Proc_{\rm CCS}$. Then
$P \sqsubseteq_{\npr} Q ~~\Leftrightarrow~~
\begin{array}[t]{@{}r@{~\supseteq~}l@{}}
\ptr(P) & \ptr(Q) \wedge \mbox{}\\
\infinite(P) & \infinite(Q) 
\end{array}$
\end{theorem}

\begin{proof}
Let $\sqsubseteq_{\it T}^\infty$ be the preorder defined by:
$P\sqsubseteq_{\it T}^\infty Q$ iff the right-hand side of \thm{nonpositive reward characterisation} holds.
\vspace{1ex}

\noindent
``$\Rightarrow$'':
Suppose $P \not\sqsubseteq_{\it T}^\infty Q$.

Case 1: Let $\sigma  \in \ptr(Q)\setminus \ptr(P)$.
Take $T:=\bar\sigma.\tau^{-1}$.
Then $T|Q$ has a computation $\pi$ with $\textit{reward}(\pi)<1$,
which $T|P$ has not. Hence $P \not\sqsubseteq_{\npr} Q$.

Case 2 proceeds exactly as in the proof of \thm{reward characterisation}.
\vspace{1ex}

\noindent
``$\Leftarrow$'': 
Suppose $P \sqsubseteq_{\it T}^\infty Q$.
Let $T\in\Tests^R_{\rm CCS}$ be a nonpositive rewards test and $r\in\IR$ be such that
there is a $\pi\in\Comp(T|Q)$ with $\textit{reward}(\pi)=r$.
It suffices to find a $\pi'\in\Comp(T|P)$ with $\textit{reward}(\pi')\leq r$.
The computation $\pi$ can be seen as a path of $T|Q$ in which all actions are $\tau$.
Decompose this path into paths $\pi_1$ of $T$ and $\pi_2$ of $Q$.
Note that $\textit{reward}(\pi)=\textit{reward}(\pi_1)$.

Moreover, $\ell(\pi_2) \in \infinite(Q) \cup\ptr(Q) \subseteq \infinite(P) \cup\ptr(P)$.
So $P$ has a path $\pi'_2$ with $\ell(\pi'_2)=\ell(\pi_2)$. Consequently, $T|P$ has an path
$\pi' \in \pi_1|\pi'_2$ that is either a computation, or a prefix of a computation, with $\textit{reward}(\pi')=r$.
In case it is a prefix of a computation $\pi''$ then $\textit{reward}(\pi'')\leq \textit{reward}(\pi')=r$.
\end{proof}
\emph{Finite-penalty nonpositive reward testing} only allows computations that incur no positive
rewards and merely finitely many penalties.  Let $P \sqsubseteq_{\rm \fpnpr} Q$ iff $\Apply(T,P)
\Sleq \Apply(T,Q)$ for every finite-penalty nonpositive reward test $T\!$.

\begin{theorem}{fp nonpositive reward characterisation}
Let $P,Q\in\Proc_{\rm CCS}$. Then
$P \sqsubseteq_{\fpnpr} Q ~~\Leftrightarrow~~
\begin{array}[t]{@{}r@{~\supseteq~}l@{}}
\ptr(P) & \ptr(Q)
\end{array}$
\end{theorem}

\begin{proof}
Let $\sqsubseteq_{\it T}$ be the preorder defined by:
$P\sqsubseteq_{\it T} Q$ iff the right-hand side of \thm{fp nonpositive reward characterisation} holds.
\vspace{1ex}

\noindent
``$\Rightarrow$'':
Suppose $P \not\sqsubseteq_{\it T} Q$.
Let $\sigma  \in \ptr(Q)\setminus \ptr(P)$.
Take $T:=\bar\sigma.\tau^{-1}$.
Then $T|Q$ has a computation $\pi$ with $\textit{reward}(\pi)<1$,
which $T|P$ has not. Hence $P \not\sqsubseteq_{\npr} Q$.
\vspace{1ex}

\noindent
``$\Leftarrow$'': 
Suppose $P \sqsubseteq_{\it T} Q$.
Let $T\in\Tests^R_{\rm CCS}$ be a finite-penalty nonpositive rewards test and $r\in\IR$ be such that
there is a $\pi\in\Comp(T|Q)$ with $\textit{reward}(\pi)=r$.
Then $\pi$ has a finite prefix $\pi^\dagger$ (not necessarily a computation) with $\textit{reward}(\pi)=r$.
It suffices to find a prefix $\pi'$ of a computation of $T|P$ with $\textit{reward}(\pi')= r$.
The finite prefix $\pi^\dagger$ can be seen as a path of $T|Q$ in which all actions are $\tau$.
Decompose this path into finite paths $\pi_1$ of $T$ and $\pi_2$ of $Q$.
Now $\ell(\pi_2) \in\ptr(Q) \subseteq \ptr(P)$.
So $P$ has a path $\pi'_2$ with $\ell(\pi'_2)=\ell(\pi_2)$. Consequently, $T|P$ has a path
$\pi' \in \pi_1|\pi'_2$ that is a prefix of a computation, with $\textit{reward}(\pi')=r$.
\end{proof}
\emph{Single penalty nonpositive reward testing} only allows computations that incur no positive
rewards and at most one penalty. Let $P \sqsubseteq_{\rm \spnpr} Q$ iff $\Apply(T,P) \Sleq \Apply(T,Q)$
for every single penalty nonpositive reward test $T\!$.
Obviously, $\sqsubseteq_{\spnpr}$ coincides with $\sqsubseteq_{\fpnpr}$.
This follows because all test used in the proof of \thm{fp nonpositive reward characterisation} have
the single penalty property.

\section{Reward may testing}

Call a test $T\in\Tests^R_{\rm CCS}$ \emph{well-behaved} if for each infinite path $T \alpha_1^{r_1} T_1 \alpha_2^{r_2} T_2 \cdots $
the limit $\lim_{n\rightarrow\infty}\sum_{i=1}^n r_i \in \IR\cup\{-\infty,\infty\}$ exists. If the
sequence $(r_i)_{i=1}^\infty$ alternates between $1$ and $-1$ for instance, the test is not well-behaved.
Since all tests used in the proof of \thm{reward characterisation} are well-behaved, the reward
testing preorder $\sqsubseteq_{\rm reward}$ would not change if one restricts the collection of
available test to the well-behaved ones only.
When restricting to well-behaved tests, the infimum $\inf_{n\rightarrow\infty}$ in
\df{reward computation} may be read as $\lim_{n\rightarrow\infty}$.

\begin{theorem}{inverse}
$P \sqsubseteq^{\rm may}_{\rm reward} Q$ iff $Q \sqsubseteq^{\rm must}_{\rm reward} P$.
\end{theorem}
\begin{proof}
For any well-behaved test $T$, let $-T$ be obtained by changing all occurrences of actions $(\alpha,r)$ into $(\alpha,\!-r)$.
Now $\Apply(-T\!,P)= \{-r \mathbin| r\mathbin\in \Apply(T\!,P)\}$. This immediately yields the claimed result.%
\end{proof}
All weaker notions of testing contemplated in \Sec{weaker} employ well-behaved tests only.
The same reasoning as above yields
(besides ${\sqsubseteq^{\rm may}_{\rm reward}} = {\sqsubseteq_{\rm reward}^{-1}}$)\\[1ex]
\mbox{}\hfill
${\sqsubseteq^{\rm may}_{\fpr}} = {\sqsubseteq_{\rm reward}^{-1}}$
\hfill,
\hfill
${\sqsubseteq^{\rm may}_{\nnr}} = {\sqsubseteq_{\npr}^{-1}}$
\hfill,
\hfill
${\sqsubseteq^{\rm may}_{\npr}} = {\sqsubseteq_{\nnr}^{-1}}$
\hfill
and
\hfill
${\sqsubseteq^{\rm may}_{\fpnpr}} = {\sqsubseteq_{\npr}^{-1}}$
\hfill.

\section{A hierarchy of testing preorders}\label{sec:hierarchy}

\begin{theorem}{must}
$P \Mustleq Q$ iff $P \sqsubseteq_{\nnr} Q$. Likewise, $P \sqsubseteq^{\rm dual}_{\rm must} Q$ iff $P \sqsubseteq_{\fpnpr} Q$.
\end{theorem}
\begin{proof}
``If'': Without affecting $\Mustleq$ one may restrict attention to tests $T\in\Tests_{\rm CCS}$ with the
property that each path of $T$ contains at most one success state---one with an outgoing transition labelled $\omega$.
Namely, any outgoing transition of a success state may safely be omitted. Now each such test $T$ can
be converted into a nonnegative reward test $T'$, namely by assigning a reward $1$ to any action leading into a success state,
keeping the rewards of all other actions $0$. The success action itself may then be renamed into $\tau$, or omitted.
Now trivially, a computation  of $T|P$ is successful iff the matching computation of $T'$
yields a reward $1$; a computation  of $T|P$ is unsuccessful iff the matching computation of $T'$
yields a reward $0$. It follows that must-testing can be emulated by nonnegative reward testing.

``Only if'': As remarked in \Sec{weaker}, nonnegative reward testing looses no power when allowing
only one reward per computation. For the same reasons it looses no power if each positive reward is $1$.
Now any reward test $T'\in\Tests^R_{\rm CCS}$ with these restrictions can be converted to a test
$T\in\Tests_{\rm CCS}$ by making any target state of a reward-1 transition into a success state.
It follows that nonnegative reward testing can be emulated by must-testing.

The second statement follows in the same way, but using a reward $-1$.
\end{proof}

\begin{figure}[h]
\input{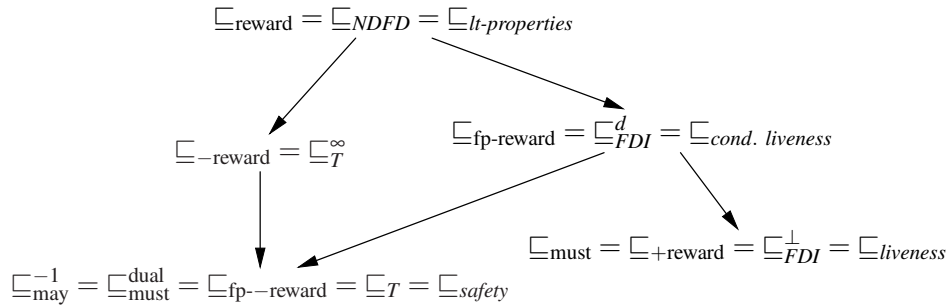}
\centerline{\raisebox{1ex}{\box\graph}}
\caption{A spectrum of testing preorders}
\label{spectrum}
\end{figure}

\noindent
A preorder $\sqsubseteq_X$ is said to be \emph{finer} than or equal to a preorder $\sqsubseteq_Y$ iff
$P \sqsubseteq_X Q \Rightarrow P \sqsubseteq_Y Q$ for all $P$ and $Q$; in that case $\sqsubseteq_Y$
is \emph{coarser} than or equal to $\sqsubseteq_X$.

\begin{theorem}{hierarchy}
The preorders occurring in this paper are related as indicated in \fig{spectrum}, where the arrows
point in the coarser direction.
\end{theorem}
\begin{proof}
The relations between $\sqsubseteq_{\rm reward}$,  $\sqsubseteq_{\fpr}$,  $\sqsubseteq_{\nnr}$,
$\sqsubseteq_{\npr}$ and $\sqsubseteq_{\fpnpr}$ follow immediately from the definitions, as the coarser
variant uses only a subset of the tests available to the finer variant.
The strictness of all these relations is obtained by the examples below.

The connections with $\Mustleq$, $\sqsubseteq^{\rm dual}_{\rm must}$ and the inverse of $\Mayleq$
are provided by \thm{must} and \Sec{dual}.
The characterisations in terms of $\sqsubseteq_{\it NDFD}$,  $\sqsubseteq_{\it FDI}^d$,
$\sqsubseteq_{\it FDI}^\bot$,  $\sqsubseteq_{\it F}^\infty$ and $\sqsubseteq_{\it T}$ are provided
by Theorems~\ref{thm:reward characterisation}--\ref{thm:fp nonpositive reward characterisation}.
The connections with $\sqsubseteq_{\textit{\scriptsize lt-properties}}$, $\sqsubseteq_{\it  cond.\ liveness}$,
$\sqsubseteq_{\it liveness}$ and $\sqsubseteq_{\it safety}$ will be established in \Sec{liveness}.
\end{proof}
  Let $a^n.P$ be defined by $a^0.P:=P$ and $a^{i+1}.P = a.a^i.P$.
  Furthermore, let \plat{$a^\infty:= \rec{\,X\defis a.X}$} be a process that performs infinitely many $a$s.
  Let $\Delta$ be the unary operator given by \plat{$\Delta P:= \rec{\,X\defis \tau.X + P~}$}. It first
performs $0$ or more $\tau$-actions, and if this number is finite subsequently behaves as its
argument $P$.
So $\Delta{\bf 0} = \tau^\infty$ just performs an infinite sequence of $\tau$-moves.

\begin{example}{fpt worse}
$\displaystyle\sum_{n\geq 1} a^n.\Delta{\bf 0} \equiv_{\fpr} a^\infty + \sum_{n\geq 1} a^n.\Delta{\bf 0}$, but
    $\displaystyle\sum_{n\geq 1} a^n.\Delta{\bf 0} \not\sqsubseteq_{\npr} a^\infty + \sum_{n\geq 1} a^n.\Delta{\bf 0}$
(and thus $\not\sqsubseteq_{\rm reward}$).\vspace{-1ex}
\end{example}

\begin{example}{fpt better}
$\Delta(c.g) \Musteq \Delta(c+c.g)$ and $\Delta(c.g) \equiv_{\npr} \Delta(c+c.g)$, yet
$\Delta(c.g) \not\sqsubseteq_{\fpr} \Delta(c+c.g)$.
These are the processes displayed in \fig{conditional}.
A test showing the latter is $c^{-1}\!.g^1$.
\end{example}

\begin{example}{must better}
$c.g \equiv_{\fpnpr} c+c.g$, yet $c.g \not\sqsubseteq_{\nnr} c+c.g$.
A test showing the latter is $c.g^1$.
\end{example}

\begin{example}{safety better}
$\Delta a \equiv_{\nnr} \Delta {\bf 0}$, yet $\Delta a \not\sqsubseteq_{\fpnpr} \Delta {\bf 0}$.
A test showing the latter is $a^{-1}$.
\end{example}

\noindent
A process $P$ is \emph{divergence-free} if $\diverg(P)=\emptyset$.
It is \emph{regular}, or \emph{finite-state}, if there only finitely many processes $Q$ such that
\plat{$\exists \sigma \in\A^*.~P \dto{\sigma} Q$}.
It is \emph{$\dto{}$-image-finite} if for each $\sigma\in A^*$ there are only finitely many $Q$ such that \plat{$P\dto{\sigma}Q$}.
Note that the class of $\dto{}$-image-finite processes is not closed under parallel composition, or under
renaming transition labels $a\in\A$ into $\tau$. Regular processes are $\dto{}$-image-finite.
Any $P\in\Proc_{\rm CCS}$ without parallel composition, relabelling or restriction is regular.
Any $P\in\Proc_{\rm CCS}$ without recursion is both divergence-free and regular.

\begin{proposition}{divergence-free}
If $P\in\Proc_{\rm CCS}$ is divergence-free, then
$P \sqsubseteq_{\nnr} Q$ iff $P \sqsubseteq_{\rm reward} Q$.
\end{proposition}

\begin{proof}
This follows immediately from Theorems~\ref{thm:reward characterisation}
and~\ref{thm:nonnegative reward characterisation}, using that
$\diverg(P)=\emptyset$, $\infd(P)=\infdd(P)=\infinite(P)$ and
$\faild(P)=\faildd(P)=\failures(P)$.
(In case $Q$ is not divergence-free one has neither $P \sqsubseteq_{\nnr} Q$ nor $P \sqsubseteq_{\rm reward} Q$.)
\end{proof}

\begin{proposition}{image-finite}
If $P$ is $\dto{}$-image-finite then (a)
$P \sqsubseteq_{\fpnpr} Q$ iff $P \sqsubseteq_{\npr} Q$\\ and (b) $P \sqsubseteq_{\fpr} Q$ iff $P \sqsubseteq_{\rm reward} Q$.
\end{proposition}

\begin{proof}
By K\"onigs lemma $\nu\in\A^\infty$ is an infinite trace of $P$ iff only if each finite prefix of $\nu$ is a partial trace of $P$.
Now (a) follows immediately from Theorems~\ref{thm:nonpositive reward characterisation}
and~\ref{thm:fp nonpositive reward characterisation}: Suppose $P \sqsubseteq_{\npr} Q$ and $\nu\in\infinite(Q)$.
Then each finite prefix of $\nu$ is in $\ptr(Q)$ and thus in $\ptr(P)$. Thus $\nu\in\infinite(P)$.

(b) follows in the same way from  Theorems~\ref{thm:reward characterisation}
and~\ref{thm:finite-penalty reward characterisation}, using \hyperlink{star}{(*)}.
\end{proof}

\section{Conditional liveness properties}\label{sec:liveness}

To obtain a general liveness property for labelled
transition systems, assume that some notion of \emph{good} is defined.
Now, to judge whether a process $P$ satisfies this liveness property,
one should judge whether $P$ can reach a state in which one would say
that something good had happened. But all observable behaviour of $P$
that is recorded in a labelled transition system until one comes to such
a verdict, is the sequence of visible actions performed until that
point. Thus the liveness property is completely determined by the set
sequences of visible actions that, when performed by $P$, lead to
such a judgement.  Therefore one can just as well define
a liveness property in terms of such a set.
\begin{definition}{liveness}
A \emph{liveness property} of processes in an LTS is given by a set
$G\subseteq \A^*$.
A process $P$ \emph{satisfies} this liveness property, notation
$P\models \textit{liveness}(G)$, when each complete trace of $P$ has a
prefix in $G$.
\end{definition}
This formalisation of liveness properties stems from \cite{vG10} and is
essentially different from the one in \cite{AS85} and most subsequent work
on liveness properties; this point is discussed in \cite[Section 6]{vG10}.

A preorder $\sqsubseteq$ \emph{preserves} liveness properties if $P \sqsubseteq Q$
implies that $Q$ enjoys any liveness property that $P$ has. It is a \emph{precongruence} for an
$n$-ary operator $\textit{op}$ if $P_i \sqsubseteq Q_i$ for $i=1,\dots,n$ implies
$\textit{op}(P_1,\dots,P_n) \sqsubseteq \textit{op}(Q_1,\dots,Q_n)$.
Now let $\sqsubseteq_{\it liveness}$ be the coarsest preorder that is a precongruence for the
operators of CSP and preserves liveness properties. In \cite{vG10} it is shown that this preorder
exists, and equals $\sqsubseteq_{FDI}^\bot$, as defined in the proof of \thm{nonnegative reward characterisation}.
The proof of this result does not require that $\sqsubseteq_{\it liveness}$ be a preorder for all
operators of CSP; it goes through already when merely requiring it to be precongruence for 
injective renaming and partially synchronous interleaving operators. Looking at this proof, the same can
also be obtained requiring $\sqsubseteq_{\it liveness}$ to be a precongruence for the CCS operators
$|$, $\backslash\RL$ and injective relabelling.

It follows that $\sqsubseteq_{\it liveness}$ coincides with $\sqsubseteq_{\nnr}$ (cf.\ \thm{hierarchy}).
This connection can be illustrated by a translation from liveness properties $G\subseteq\A^*$
(w.l.o.g.\ assumed to have the property that if $\sigma\in G$ then $\sigma\rho\notin G$ for any $\rho\neq\epsilon$)
to nonnegative reward tests $T_G$. Here $T_G$ can be rendered as a deterministic tree
in which all transitions completing a trace from $\bar G$ yield a reward $1$, so that all computations of $T|P$ earn
a positive reward iff $P\models \textit{liveness}(G)$.

One obtains a general concept of safety property by means of the same argument as for liveness
properties above, but using ``bad'' instead of ``good''.

\begin{definition}{safety}
A \emph{safety property} of processes in an LTS is given by a set
$B\subseteq \A^*$.
A process $P$ \emph{satisfies} this safety property, notation
$P\models \textit{safety}(B)$, when $\ptr(p)\cap B=\emptyset$.
\end{definition}
This formalisation of safety properties stems from \cite{vG10} and is in line with the one in \cite{AS85}.
Now let $\sqsubseteq_{\it safety}$ be the coarsest precongruence (for the same choice of operators as above) that
preserves safety properties. In \cite{vG10} it is shown that this preorder
exists, and equals $\sqsubseteq_{T}$, as defined in the proof of \thm{fp nonpositive reward characterisation}.

It follows that $\sqsubseteq_{\it safety}$ coincides with $\sqsubseteq_{\fpnpr}$ (cf.\ \thm{hierarchy}).
This connection can be illustrated by a translation from safety properties $B\subseteq\A^*$
(w.l.o.g.\ assumed to have the property that if $\sigma\in B$ then $\sigma\rho\notin B$ for any $\rho\neq\epsilon$)
to nonnegative reward tests $T_B$. Here $T_B$ can be rendered as a deterministic tree
in which all transitions completing a trace from $\bar B$ yield a reward $-1$, so that all computations of $T|P$ earn
a nonnegative reward iff $P\models \textit{safety}(B)$.

A conditional liveness property says that \emph{under certain conditions something good will eventually happen}.
To obtain a general conditional liveness property for LTSs,
assume that some condition, and some notion of \emph{good} is defined.
Now, to judge whether a process $P$ satisfies this conditional liveness property,
one should judge first of all in which states the condition is fulfilled.
All observable behaviour of $P$ that is recorded in an LTS
until one comes to such a verdict, is the sequence of visible actions
performed until that point. Thus the condition is completely
determined by the set of sequences of visible actions that, when
performed by $P$, lead to such a judgement. Next one should judge
whether $P$ can reach a state in which one would say that something
good had happened. Again, this judgement can be expressed in terms of
the sequences of visible actions that lead to such a state.

\begin{definitionA}{\cite{vG10}}{conditional liveness}
A \emph{conditional liveness property} of processes in an LTS is given
by two sets $C,G\subseteq \A^*$.
A process $P$ \emph{satisfies} this conditional liveness property, notation
$P\models \textit{liveness}_C(G)$, when each complete trace of
$P$ that has a prefix in $C$, also has a prefix in $G$.
\end{definitionA}
Now let $\sqsubseteq_{\it cond.\ liveness}$ be the coarsest precongruence (for the same choice of operators as above) that
preserves conditional liveness properties. In \cite{vG10} it is shown that this preorder
exists, and equals $\sqsubseteq_{FDI}^d$, as defined in the proof of \thm{finite-penalty reward characterisation}.\pagebreak[2]
It follows that $\sqsubseteq_{\it safety}$ coincides with $\sqsubseteq_{\fpr}$ (cf.\ \thm{hierarchy}).
Similar to the above cases, this connection can be illustrated by a translation from conditional
liveness properties $C,G\subseteq\A^*$ to reward tests in which each computation has at most one
negative and one positive reward, which are always $-1$ and $+1$.

\begin{definition}{LT}
A \emph{linear time property} of processes in an LTS is given by a set
$\Phi\subseteq \A^*\cup\A^\infty$ of finite and infinite sequences of actions.
A process $P$ \emph{satisfies} this property, notation
$P\models \Phi$, when $\ct(P)\subseteq \Phi$.
\end{definition}
A liveness property is a special kind of linear time property:\\
\mbox{}\hspace{1.6cm}$\emph{liveness}(G)= \{\sigma\in\A^*\cup\A^\infty \mid \exists \rho\in G.\; \rho\leq\sigma\}$.\\
Likewise,
$\emph{safety}(B) =\{\sigma\mathbin\in\A^*\cup\A^\infty \mid \neg \exists \rho\mathbin\in B.\; \rho\leq\sigma\}$,
and\\ \mbox{}\hspace{1.6cm}$\emph{liveness}_C(G)=
\{\sigma\in\A^*\cup\A^\infty \mid (\exists \rho\in C.\; \rho\leq\sigma)
\Rightarrow (\exists \nu\in G.\; \nu\leq\sigma)\}$.

Now let $\sqsubseteq_{lt.\ properties}$ be the coarsest precongruence (for the same choice of operators as above) that
preserves linear time properties. In \cite{KV92,vG10} it is shown that this preorder
exists, and equals $\sqsubseteq_{NDFD}$, as defined in the proof of \thm{reward characterisation}.
It follows that $\sqsubseteq_{\it lt.\ properties}$ coincides with $\sqsubseteq_{\rm reward}$ (cf.\ \thm{hierarchy}).

\section{Congruence properties}

\begin{theorem}{congruence}
The preorders of this paper are precongruences for the CCS operators $|$, $\backslash\RL$ and $[f]$.
\end{theorem}

\begin{proof}
Note that $\Apply(T,R|P)=\Apply(T|R,P)$, using the associativity (up to strong bisimilarity) of~$|$.
Therefore $P \sqsubseteq_{\rm reward} Q$ implies $R|P \sqsubseteq_{\rm reward} R|Q$, showing that
$\sqsubseteq_{\rm reward}$ is a precongruence for parallel composition. The same holds for
$\sqsubseteq_{\fpr}$, $\sqsubseteq_{\nnr}$, $\sqsubseteq_{\npr}$ and $\sqsubseteq_{\fpnpr}$.

Likewise  $\Apply(T,P\backslash\RL)=\Apply(T\backslash\RL,P)$. This yields precongruence results for restriction.

Finally, $\Apply(T,P[f])=\Apply(T[f^{-1}],P)$, yielding  precongruence results for relabelling.\\
Here $[f^{-1}]$ is an operator with rule $\displaystyle\frac{E \goto{\alpha,r} E'}{E[f^{-1}] \goto{\beta,r} E'[f^{-1}]}~~(f(\beta)=\alpha)$.
Although this is not a CCS operator, for any test $T$ the test $T[f^{-1}]$ is expressible in CCS,
on grounds that each process in an LTS is expressible in CCS\@.
\end{proof}

\begin{theorem}{action prefixing}
The preorders of this paper are precongruences for action prefixing.
\end{theorem}

\begin{proof}
This follows in a straightforward way from the characterisations of the preorders in
Sections~\ref{sec:explicit} and~\ref{sec:weaker}.
For instance, $\faildd(a.P)=\{\langle a\sigma,X\rangle \mid \langle \sigma,X\rangle \in \faildd(a.P)\}$.
\end{proof}
In the same way it follows that $\sqsubseteq_{\fpnpr}$ and $\sqsubseteq_{\npr}$ are precongruences
for the CCS operator $+$.
However, the preorders $\sqsubseteq_{\rm reward}$, $\sqsubseteq_{\fpr}$ and $\sqsubseteq_{\nnr}$
fail to be congruences for choice:

\begin{example}{choice congruence}
  $ {\bf 0} \equiv_{\rm reward} \tau$, yet ${\bf 0}+a \not\sqsubseteq_{\nnr} \tau+a$,
  using that $\langle\epsilon,\A\rangle\in\faild(\tau+a)\setminus\faild({\bf 0}+a)$.
\end{example}
This issue occurs for almost all semantic equivalences and preorders that abstract from internal
actions. The standard solution is to replace each such preorder $\sqsubseteq_X$
by the coarsest precongruence for the operators of CCS that is finer than $\sqsubseteq_X$.
Let $\textit{stable}$ be the predicate that holds for a process $P$ iff there is no $P'$ with $P\ar{\tau}P'$.
Write $P\sqsubseteq^\tau_X Q$ iff $P\sqsubseteq_X Q \wedge (\textit{stable}(P) \mathbin\Rightarrow \textit{stable}(Q))$.

\begin{theorem}{closure}
  Let $X\in\{\scriptstyle{\rm reward},~\fpr,~\nnr\}$. Then $\sqsubseteq^\tau_X$
  is the coarsest precongruence for the operators of CCS that is contained in $\sqsubseteq_X$.
\end{theorem}
\begin{proof}
  That $\sqsubseteq^\tau_{\nnr}$ is a precongruence for $+$ follows with \thm{nonnegative reward characterisation} since
  \[\begin{array}[t]{@{}r@{~}l@{}}
\textit{stable}(P+Q) ~\Leftrightarrow& ~\textit{stable}(P) \wedge \textit{stable}(Q)\\[1ex]
 \faild(P+Q) = &  \{\langle \sigma,X \rangle \in \faild(P)\mid \sigma\neq\epsilon \vee \neg\textit{stable}(P)\} \cup \mbox{}\\
     & \{\langle \sigma,X \rangle \in \faild(Q)\mid \sigma\neq\epsilon \vee \neg\textit{stable}(Q)\} \cup \mbox{} \\
     & \{\langle \epsilon,X \rangle \mid \langle \epsilon,X \rangle\in\faild(P) \cap \faild(Q)\},\\[1ex]
       \infd(P+Q) ~=& \infd(P) \cup \infd(Q)\\
       \divd(P+Q) ~=& \divd(P) \cup \divd(Q)\;.
    \end{array}
  \]
  That it is a congruence for action prefixing, $|$, $\backslash\RL$ and $[f]$ follows since
  \[\begin{array}{@{}r@{~~\mbox{iff}~~}l@{}}
  \textit{stable}(\alpha.P) & \alpha\neq\tau\\
  \textit{stable}(P|Q) & \textit{stable}(P) \wedge \textit{stable}(Q) \wedge {\neg\exists}
  a\in\A\!.~(\langle a,\emptyset\rangle \in \faild(P) \wedge \langle \bar a,\emptyset\rangle \in \faild(P))\\
  \textit{stable}(P\backslash\RL) & \textit{stable}(P)\\
  \textit{stable}(P[f]) & \textit{stable}(P).\end{array}\]
  By definition, $\sqsubseteq_{\nnr}^\tau$ is contained in $\sqsubseteq_{\nnr}$.
  To see that it is the coarsest precongruence contained in $\sqsubseteq_{\nnr}$, suppose
  $P\not\sqsubseteq_{\nnr}^\tau Q$. It suffices to build a context $C[\_\!\_]$ from CCS operators such
  that $C[P]\not\sqsubseteq_{\nnr} C[Q]$. The case $P\not\sqsubseteq_{\nnr} Q$ is immediate---take the
  trivial context with $C[P]:=P$. So assume $P\sqsubseteq_{\nnr} Q$.
  Then $\textit{stable}(P)$ and $\neg\textit{stable}(Q)$.
  Hence $\epsilon\notin\divd(P)\supseteq \divd(Q)$. 
  Choose $a\notin\ptr(Q)$---in case no such $a$ exists, one first applies an injective relabelling
  to $P$ and $Q$ such that $a\not\in\mbox{range}(f)$.
  Now $\langle\epsilon,\{a\}\rangle\in\failures(Q)\subseteq \faild(Q) \subseteq \faild(P)$.
  However, whereas $\langle\epsilon,\{a\}\rangle\in \faild(Q+a)$
  one has $\langle\epsilon,\{a\}\rangle\notin\faild(P+a)$.
  It follows that $P+a \not\sqsubseteq_{\nnr} Q+a$.

  The arguments for $X\in\{{\scriptstyle\rm reward},~\fpr\}$ are very similar.
\end{proof}

\section{Axiomatisations}\label{sec:axiomatisations}

The following axioms are easily seen to be sound for $\sqsubseteq_{\rm reward}^\tau$.
Here an equality $P\equiv Q$ can be seen as a shorthand for the two axioms $P\sqsubseteq Q$ and $Q \sqsubseteq P$.
Action prefixing and $\Delta$ bind stronger than $+$.
\[\begin{array}{clrclc}
 \multirow{3}{*}{$\left\{\rule{0pt}{21pt}\right.$}
& \mylabel{R1} &  \tau.X + Y &\equiv& \tau.X + \tau.(X+Y) \\
& \mylabel{R2} &  \alpha.X + \tau.(\alpha.Y+Z) &\equiv& \tau(\alpha.X + \alpha.Y + Z) \\
& \mylabel{R3} &  \alpha.(\tau.X + \tau.Y) &\equiv& \alpha.X + \alpha.Y & \multirow{3}{*}{$\left.\rule{0pt}{21pt}\right\}$}\\
& \mylabel{RP1} &  \tau.X + Y &\sqsubseteq& \tau.(X+Y) \\
& \mylabel{RP2} & \tau.X + Y &\sqsubseteq& X \\
& \mylabel{R4} & \tau.\Delta X + Y &\equiv& \Delta (X+Y)
\end{array}\]
For recursion-free processes, and dropping the infinite choice operator in favour of $+$ and ${\bf 0}$,
$\Mustleq^\tau$ coincides with $\sqsubseteq_{\rm reward}^\tau$ and $\sqsubseteq_{\fpr}^\tau$.
Together with the standard axioms for strong bisimilarity \cite{ccs}, the three axioms \myref{R1}--\myref{R3}
constitute a sound and complete axiomatisation of $\Musteq^\tau$ \cite[Theorem 4.2]{CFG15},\pagebreak[1] and thus for $\equiv_{\rm reward}^\tau$.
Likewise, the three axioms \myref{RP1},\myref{RP2} and \myref{R3} constitute a sound and complete
axiomatisation of $\Mustleq^\tau$ \cite[Theorem 4.1]{CFG15}, and thus for $\sqsubseteq_{\rm reward}^\tau$;
the axioms \myref{R1} and \myref{R2} are derivable from them. The first sound and complete
axiomatisation of $\Mustleq^\tau$ appears in \cite{DH84}; their axioms are derivable from the ones
above (and vise versa).

A sound and complete axiomatisation of $\Mayeq$ (and hence of $\equiv_{\npr}$) is obtained by adding
the axioms $\tau.X\equiv X$ and $\alpha(X+Y) \equiv \alpha.X + \alpha.Y$ to the standard axioms for
strong bisimilarity \cite[Theorem 4.5]{CFG15}. The axioms \myref{R1}--\myref{R3} are derivable from them.
Adding the axiom $X+Y\sqsubseteq X$ yields a sound and complete axiomatisation of $\Mayleq^{-1}$
(and hence of $\sqsubseteq_{\npr}$) \cite[Theorem 4.6]{CFG15}. The axioms \myref{RP1} and
\myref{RP2} are then also derivable. The first sound and complete axiomatisation of $\Mayleq$
appears in \cite{DH84}; their axioms are derivable from the ones above (and vise versa).

To illustrate the difference between $\Musteq^\tau$ and $\equiv_{\rm reward}^\tau$, without having
to deal with recursion, I consider recursion-free CCS with finite choice (as done above),
but upgraded with the \emph{delay operator} $\Delta$ introduced in \cite{BKO87} and in \Sec{hierarchy}.
Clearly all preorders of this paper are precongruences for $\Delta$.
With \myref{R4}, sound for $\equiv_{\rm reward}^\tau$, one can derive $\tau.\Delta X \equiv \Delta X$
and $\Delta X + Y \equiv \Delta (X+Y)$. Writing $\Omega$ for $\Delta {\bf 0}$, the latter implies
$\Delta Y \equiv \Omega+Y$ so one can equally well take $\Omega$ as $\Delta$ as primitive. 
It also follows that $\Delta\Delta X \equiv \Delta X$.

The above sound and complete axiomatisations of $\Mayeq$ and $\Mayleq^{-1}$ (and hence of $\equiv_{\npr}$ and $\sqsubseteq_{\npr}$) 
are extended with $\Delta$ by adding the trivial axiom $\Delta X=X$; \myref{R4} is then derivable.
This illustrates that these preorders abstract from divergence.
The axiom
\[\mylabel{R5} \qquad \Delta X \equiv \Delta Y\]
is sound for $\Musteq^\tau$. It expresses that must testing does not record any information past a divergence.
Axioms \myref{RP2}, \myref{R4} and \myref{R5} imply $\Omega \sqsubseteq X$, an axiom featured in \cite{DH84}.
Neither $\Delta X=X$ nor \myref{R5} is sound for $\equiv_{\rm reward}^\tau$.

\section{Failure of congruence property for recursion}

Each preorder $\sqsubseteq$ on CCS processes (= closed CCS expressions) can be extended to one on all CCS
expressions by defining $E \sqsubseteq F$ iff all closed substitution instances of this inequality hold.

\begin{definition}{full}
A preorder $\sqsubseteq$ on $\IE_{\rm CCS}$ is a (full) precongruence for recursion if
$S_Y \sqsubseteq T_Y$ for each $Y\in\dom(S)=\dom(T)$ implies $\rec{S} \subseteq \rec{T}$.
\end{definition}
The following counterexample shows that the must-testing preorder $\Mustleq^\tau$ fails to be a
precongruence for recursion, implying that the must-testing equivalence $\Musteq^\tau$ fails to be a
congruence for recursion.

\begin{example}{recursion congruence}
Let $P\in\IT_{\rm CCS}$ be such that $\epsilon\notin\diverg(P)$---for instance $P={\bf 0}$.\vspace{-3pt}
Then by \myref{R1} one has $\tau.P + X \Musteq^\tau \tau.P + \tau.(X+P)$.
Yet $\rec{\,X\defis \tau.P + X} \not\Mustleq^\tau \rec{\,X\defis \tau.P + \tau.(X+P)}$,
because only the latter process has a divergence $\epsilon$.
\end{example}
The same example shows that also $\sqsubseteq_{\rm reward}^\tau$,  $\sqsubseteq_{\fpr}^\tau$,
$\sqsubseteq_{\rm reward}$,  $\sqsubseteq_{\fpr}$ and $\Mustleq$ fail to be precongruences for recursion.
However, I conjecture that all these preorders are \emph{lean} precongruences for recursion as
defined in \cite{vG17b}.

\section{Unguarded recursion}

The must-testing preorder $\Mustleq$ on CCS presented in this paper is not quite the same as the original
one $\Mustleq^{\rm org}$ from \cite{DH84}. The following example shows the difference.

\begin{example}{unguarded recursion}
\hfill ${\bf 0} \qquad \begin{array}{c}\Musteq\\[1ex]\not\Mustleq^{\rm org}\end{array} \qquad
\rec{\,X\defis X} \qquad \begin{array}{c}\not\Mustleq\\[1ex] \Musteq^{\rm org} \end{array} \qquad \rec{\,X\defis \tau.X}$.\hfill\mbox{}
\end{example}
The $\Musteq$-statement follows since neither process has a single outgoing transition;\vspace{1pt} the
processes are even \emph{strongly bisimilar} \cite{ccs}.
The $\not\Mustleq$-statement follows since \plat{$\epsilon \in \diverg(\rec{\,X\defis \tau.X})$},\vspace{1pt} yet \plat{$\epsilon\notin\diverg(\rec{\,X\defis X})$}.
A test showing the difference is $\tau.\omega$.

The reason that in the original must-testing approach $\rec{\,X\defis X}$ sides with $\rec{\,X\defis\tau.X}$\vspace{2pt}
rather than with ${\bf 0}$,
is that \cite{DH84} treats a process featuring unguarded recursion (cf.\ \cite{ccs}), such as \plat{$\rec{\,X\defis X}$}, as
if it diverges, regardless whether it can do any internal actions $\tau$.
This leads to a must-testing equivalence that is incomparable with strong bisimilarity.

In my view, the decision whether \plat{$\rec{\,X\defis X}$} diverges or not is part of the
definition of the process algebra CCS, and entirely orthogonal to the development of testing equivalences.
Below I define a process algebra CCS$_\bot$ that resembles CCS in all aspects, expect that
any process with unguarded recursion is declared to diverge. I see the work of \cite{DH84} not so
much as defining a must-testing equivalence on CCS that is incomparable with strong bisimilarity, but
rather as defining a must-testing equivalence on CCS$_\bot$, a languages that is almost, but not
quite, the same as CCS\@.\footnote{All processes of \ex{unguarded recursion} are \emph{weakly bisimilar} \cite{ccs}.
In my view this does not mean that weak bisimulation semantics uses a variant of CCS in which none
of these processes diverges. Instead it tells that weak bisimilarity abstracts from divergence.}
This is a matter of opinion, as there is no technical difference between these approaches.

I now proceed to define CCS$_\bot$, and apply the reward testing preorders of this paper to that language.

\begin{definition}{bot}
Let $\downarrow$ be the least predicate on $\Proc_{\rm CCS}$ which satisfies\vspace{-1ex}
\begin{itemise}
\item $\alpha.P \mathop{\downarrow}$ for any $\alpha\in Act$,
\item if $P_i \mathop{\downarrow}$ for all $i\in I$ then $\sum_{i\in I}P_i \mathop{\downarrow}$,
\item if $P \mathop{\downarrow}$ and $Q \mathop{\downarrow}$ then $P|Q \mathop{\downarrow}$, $P\backslash\RL  \mathop{\downarrow}$
  and $P[f] \mathop{\downarrow}$,
\item if $\rec[S_X]{S} \mathop{\downarrow}$ then $\rec{S} \mathop{\downarrow}$.
\end{itemise}
Let $P \mathop{\uparrow}$ if not $P \mathop{\downarrow}$. If $P \mathop{\uparrow}$ then $P$
features \emph{strongly unguarded recursion}.\footnote{Un(strongly unguarded) recursion should not
  be called ``strongly guarded'' recursion; it is weaker than guarded recursion.}
\end{definition}
Note that ${\bf 0} \mathop{\downarrow}$, $\rec{\,X\defis X}  \mathop{\uparrow}$ and $\rec{\,X\defis \tau.X} \mathop{\downarrow}$,
the latter because in \df{bot} $\tau$ is allowed as a \emph{guard}.
The definitions of this paper are adapted to CCS$_\bot$ by redefining 
 $P$ \emph{diverges}, notation $P{\Uparrow}$, if either there is a $P'$ with $P \dto{} P' \mathop{\uparrow}$
or there are $P_i\in\mathbb{P}$ for all $i>0$ such that \plat{$P\goto{\tau}P_1\goto{\tau}P_2\goto{\tau} \cdots$}.
In \df{reward computation}, and similarly for \df{computation}, clause (i) is replaced by (i$'$) ``if $T_n$
is the final element in $\pi$, then either $T_n\mathop{\uparrow}$ or \plat{$T_n \ar{\tau,r} T$} for no $r$ and $T$''.
Now all results for CCS from \Sects~\ref{sec:classical}--\ref{sec:axiomatisations} remain valid for
CCS$_\bot$ as well. The only change in the proofs of Theorems~\ref{thm:reward characterisation}--\ref{thm:nonnegative reward characterisation},
direction ``$\Leftarrow$'', is that finite paths ending in $\downarrow$ are treated like infinite paths.

My definition of $\Mustleq$ on CCS$_\bot$ differs on two points from the definition of $\Mustleq$ on
CCS$_\bot$ from \cite{DH84}. But both differences are inessential, and the resulting notion of
$\Mustleq$ is the same. The first difference is that in \cite{DH84} the notion of computation is
exactly as in \df{computation}, rather than the amended form above. However, in \cite{DH84} a
computation $\pi=T_0,T_1,T_2,\dots\in\Comp(T|P)$ counts as successful only if
(a) it contains a state $T$ with $T \ar{\omega} T'$ for some $T'$,
and (b) if $T_k\mathop\uparrow$ then $T_{k'} \ar{\omega} T'$ for some $T'$ and some $k'\leq k$.
It is straightforward to check that $\Apply(T|P)$ remains the same upon dropping (b) and changing (i) into (i$'$).\linebreak[2]
The other difference is that in \cite{DH84} $\tau$ does not count as a guard---their version of \df{bot} requires $\alpha\in\A$.
So in \cite{DH84} one has \plat{$\rec{\,X\defis \tau.X} \mathop{\uparrow}$}. The notion of $\downarrow$
from \cite{DH84} is therefore closer to unguarded recursion rather than strongly unguarded recursion.
However, in the treatment of \cite{DH84} one would have \plat{$\rec{\,X\defis a.X|\bar{a}\,} \mathop{\downarrow}$},
showing that the resulting notion of guardedness is not very robust.
Since the essential difference between CCS and CCS$_\bot$ is that in CCS$_\bot$ a strongly unguarded
recursion is treated as a divergence, it does not matter whether $\downarrow$ also includes all or
some not-strongly unguarded recursions, such as \plat{$\rec{\,X\defis \tau.X}$}. For any such
not-strongly unguarded recursion is already divergent, and hence it does not make difference whether
it is declared syntactically divergent as well.

An alternative to moving from CCS to CCS$_\bot$ is to restrict either language to processes $P$
satisfying $P\mathop\downarrow$. This restriction rules out the process \plat{$\rec{\,X\defis X}$},
but includes \plat{$\rec{\,X\defis \tau.X}$}. On this restricted set of processes their is no
difference between CCS and CCS$_\bot$.

Another approach to making unguarded recursions divergent is to change the rule \myref{Rec} from
\tab{CCS} into $\rec{S}\ar\tau\rec[S_X]{S}$; this is done in the setting of CSP \cite{OH86}.
This would not have the same result, however, as here and in \cite{DH84} one has \plat{$a+\rec{X\defis b}\Musteq a+b$}.

The great advantage of moving from CCS to CCS$_\bot$ is that Counterexample~\ref{ex:recursion congruence}, against
testing preorders being congruences for recursion, disappears.
\begin{open}{recursion congruence}
  Are $\sqsubseteq_{\rm reward}^\tau$,  $\sqsubseteq_{\fpr}^\tau$ and $\Mustleq^\tau$ precongruences
  for recursion on CCS$_\bot$?
\end{open}
In \cite{DH84} it is shown that, in the absence of infinite choice, $\Mustleq^\tau$ is a precongruences for recursion.
Central in the proof is that on CCS$_\bot$ with finite choice, the clause on infinite traces ($\infd(P)\supseteq\infd(Q)$)
may be dropped from \thm{nonnegative reward characterisation}, since the infinite traces
$\infd(P)$ of a CCS$_\bot$ process $P$ with finite choice are completely determined by $\divd(P)$ and $\faild(P)$.
This proof does not generalise to $\sqsubseteq_{\rm reward}^\tau$ or $\sqsubseteq_{\fpr}^\tau$,
since here, on CCS$_\bot$ with finite choice, the infinite traces are not redundant.
The proof also does not generalise to $\Mustleq^\tau$ on CCS with infinite choice.

In \cite{Ros97} it is shown that $\sqsubseteq_{\it FDI}^\bot$ (cf.\ \thm{nonnegative reward characterisation}),
which coincides with $\Mustleq$, is a congruence for recursion on the language CSP\@.
I expect that similar reasoning can show that $\sqsubseteq_{\rm reward}^\tau$ is a congruence for recursion on CCS$_\bot$.
In \cite{Ros05} it is shown that $\sqsubseteq_{\it FDI}^d$ (cf.\ \thm{finite-penalty reward characterisation}),
which coincides with $\sqsubseteq_{\fpr}$,  is a congruence for recursion on CSP\@.
I expect that similar reasoning can show that $\sqsubseteq_{\fpr}^\tau$ is a congruence for recursion on CCS$_\bot$.
Roscoe \cite{Ros05} also presents an example, independently discovered by Levy \cite{Levy08},
showing that $\equiv_{\it NDFD}$ (cf.\ \thm{reward characterisation}),
which coincides with $\sqsubseteq_{\rm reward}$, fails to be a congruence for recursion:\footnote{The
  example was formulated for another equivalence, but actually applies to a range
  of equivalences, including $\equiv_{\it NDFD}$.}
Let $\textit{FA}$ be a process that has \emph{all} conceivable failures, divergences and infinite
traces, except for the infinite trace $a^\infty$. Then $FA+\tau.X \equiv_{\it NDFD} FA+a.X$,
for both sides have all conceivable failures, divergences and infinite
traces, with the possible exception of $a^\infty$, and both side have the infinite trace $a^\infty$
iff $X$ has it. However,\vspace{-2ex}
$$\rec{FA+\tau.X} ~\not\equiv_{\it NDFD}~ \rec{FA+a.X}\vspace{-1ex}$$
since only the latter process has the infinite trace $a^\infty$.

It could be argued that this example shows that the definition of being a congruence for recursion
ought to be sharpened, for instance by requiring that $E \sqsubseteq F$ holds only if all
closed substitutions of $E \sqsubseteq F$ employing an extended alphabet of actions hold.
This would invalidate $FA+\tau.X \equiv_{\it NDFD} FA+a.X$, namely by substituting $b$ for $X$, with
$b$ a fresh action, not alluded to in $FA$. With such a sharpening, the question whether
$\sqsubseteq_{\rm reward}^\tau$ is a congruence for recursion on CCS$_\bot$ is open.

\section{Related work}

The concept of reward testing stems from \cite{JHW94}, in the setting of nondeterministic
probabilistic processes. In the terminology of \Sec{weaker}, they employ single reward nonnegative
reward testing.  In \cite{DGMZ07} it was shown, again in a probabilistic setting, that nonnegative
reward testing is no more powerful then classical testing. This result is a
probabilistic analogue of \thm{must}.  Negative rewards were first proposed in \cite{vG09}, a
predecessor of the present paper.\pagebreak[2]  In \cite{DGHM13}, reward testing with also negative rewards,
called \emph{real-reward} testing, was applied to nondeterministic probabilistic processes.
Although technically no rewards can be gathered after a first reward has been encountered, thanks to
probabilistic branching rewards can be distributed over multiple actions in a computation. This
makes the approach a probabilistic generalisation of the reward testing proposed here. The main
result of \cite{DGHM13} is that for finitary (= finite-state and finitely many transitions)
nondeterministic probabilistic processes without divergence, real-reward testing coincides with
nonnegative reward testing. This is a generalisation (to probabilistic processes) of a specialisation
(to finitary processes) of \pr{divergence-free}. An explicit characterisation (as in
\thm{reward characterisation}) of real-reward testing for processes with divergence was not
attempted in \cite{DGHM13}.

The \emph{nondivergent failures divergences} equivalence, $\equiv_{\it NDFD}$, defined in the proof
of \thm{reward characterisation}, stems from \cite{KV92}. There it was shown to be
the coarsest congruence (for a collection of operators equivalent to the ones used in \Sec{liveness})
that preserves those linear-time properties (cf.\ \df{LT}) that can be expressed in linear-time
temporal logic without the nexttime operator. If follows directly from their proof that
it is also the  coarsest congruence that preserves \emph{all} linear-time properties as defined in \df{LT};
so $\equiv_{\it NDFD}$ coincides with $\equiv_{\it lt.\ properties}$, as remarked at the end of \Sec{liveness}.
It is this result that inspired \thm{reward characterisation} in the current paper.

The paper \cite{Le94} argues that $\equiv_{\it NDFD}$ can be seen as a testing equivalence, but does not offer
a testing scenario in quite the same style as \cite{DH84} or the current paper.

The semantic equivalence $\equiv_{\it FDI}^d$, whose associated preorder occurs in the proof of 
\thm{finite-penalty reward characterisation}, stems from \cite{Pu01}.  There it was shown to be
the coarsest congruence (for the same operators)
that preserves $\deadlocks(P) \cup \diverg(P)$, the combined deadlock and divergence traces of a
process (cf.\ \df{traces}). It is this result that directly led (via \cite[Theorem 9]{vG10})
to \thm{finite-penalty reward characterisation} in the current paper.

In \cite{DH84} the action $\omega$ is used merely to mark certain states as success states, namely
the states were an $\omega$-transition is enabled; a computation is successful iff it passes through
such a success state. In \cite{Seg96}, on the other hand, it is the \emph{actual execution} of
$\omega$ that constitutes success. In \cite{DGMZ07,DGHM08}, this is called \emph{action-based testing};
\cite[Proposition 5.1 and Example 5.3]{DGHM08} shows that action-based must testing is strictly
less discriminating than state-based must-besting:
\begin{center}
$\tau.a.\Omega \Musteq^{\mbox{\scriptsize action-based}} \tau.a.\Omega + \tau.{\bf 0}$,
\qquad whereas \qquad $\tau.a.\Omega \not\Mustleq \tau.a.\Omega + \tau.{\bf 0}$.
\end{center}
The preorders in the current paper are generalisations of state-based testing;
an action-based form of reward testing could be obtained by only allowing $\tau$-actions to carry
non-0 rewards. The same counterexample as above would show the difference between state- and
action-based reward testing.

The reward testing contributed here constitutes a strengthening of the testing machinery of De
Nicola \& Hennessy. As such it differs from testing-based approaches that lead to incomparable
preorders, such as the \emph{efficiency testing} of \cite{Vo02}, or the \emph{fair testing}
independently proposed in \cite{BRV95} and \cite{NC95}.

In \cite{vG16} I advocate an overhaul of concurrency theory to ensure liveness properties
when making the reasonable assumption of \emph{justness}. The current work is prior to any such overhaul.
It is consistent with the principles of \cite{vG16} when pretending that the parallel composition $|$ of
CCS is in fact not a parallel composition of independent processes, but an interleaving operator,
scheduling two parallel treads by means of arbitrary interleaving.

\section{Conclusion}

In this paper I contributed a concept of reward testing, strengthening the may and must testing of
De Nicola \& Hennessy. Inspired by \cite{KV92,Pu01}, I provided an explicit characterisation of the
reward-testing preorder, as well as of a slight weakening, called finite-penalty reward
testing. Must testing can be recovered by only considering positive rewards, and may testing by only
considering negative rewards.  While the must-testing preorder preserves liveness properties, and
the inverse of the may-testing preorder (which can also be seen as a must-testing preorder dealing
with catastrophes rather than successes) preserves safety properties, the (finite-penalty) reward
testing preorder, which is finer than both, additionally preserves conditional liveness
properties. I illustrated the difference between may testing, must testing and (finite-penalty)
reward testing in terms of their equational axiomatisations.  When applied to CCS as intended by
Milner, must-testing equivalence fails to be a congruence for recursion, and the same problem exists
for reward testing. The counterexample is eliminated by applying it to a small variant of CCS that,
following \cite{DH84}, treats a process with unguarded recursion as if it is diverging, even if it
cannot make any internal moves. In this setting, by analogy with Roscoe's work on CSP \cite{Ros97,Ros05},
I expect must-testing and finite-penalty reward testing to be congruences for recursion; for reward
testing this question remains open.

\bibliographystyle{eptcsini}
\bibliography{gdpf}

\begin{thebibliography}{10}
\providecommand{\bibitemdeclare}[2]{}
\providecommand{\surnamestart}{}
\providecommand{\surnameend}{}
\providecommand{\urlprefix}{Available at }
\providecommand{\url}[1]{\texttt{#1}}
\providecommand{\href}[2]{\texttt{#2}}
\providecommand{\urlalt}[2]{\href{#1}{#2}}
\providecommand{\doi}[1]{doi:\urlalt{http://dx.doi.org/#1}{#1}}
\providecommand{\bibinfo}[2]{#2}

\bibitemdeclare{incollection}{AJ94}
\bibitem{AJ94}
\bibinfo{author}{S.~\surnamestart Abramsky\surnameend} \&
  \bibinfo{author}{A.~\surnamestart Jung\surnameend} (\bibinfo{year}{1994}):
  \emph{\bibinfo{title}{Domain Theory}}.
\newblock In: {\sl \bibinfo{booktitle}{Handbook of Logic and Computer
  Science}}, \bibinfo{volume}{3}, \bibinfo{publisher}{Clarendon Press}, pp.
  \bibinfo{pages}{1--168}.

\bibitemdeclare{article}{AS85}
\bibitem{AS85}
\bibinfo{author}{B.~\surnamestart Alpern\surnameend} \&
  \bibinfo{author}{F.B.~\surnamestart Schneider\surnameend}
  (\bibinfo{year}{1985}): \emph{\bibinfo{title}{Defining liveness}}.
\newblock {\sl \bibinfo{journal}{Information Processing Letters}}
  \bibinfo{volume}{21}(\bibinfo{number}{4}), pp. \bibinfo{pages}{181--185},
  \doi{10.1016/0020-0190(85)90056-0}.

\bibitemdeclare{inproceedings}{BKO87}
\bibitem{BKO87}
\bibinfo{author}{J.A.~\surnamestart Bergstra\surnameend},
  \bibinfo{author}{J.W.~\surnamestart Klop\surnameend} \&
  \bibinfo{author}{E.-R.~\surnamestart Olderog\surnameend}
  (\bibinfo{year}{1987}): \emph{\bibinfo{title}{Failures without chaos: a new
  process semantics for fair abstraction}}.
\newblock In \bibinfo{editor}{M.~\surnamestart Wirsing\surnameend}, editor:
  {\sl \bibinfo{booktitle}{Formal Description of Programming Concepts -- III,
  Proceedings of the $3^{th}$ IFIP WG 2.2 working conference, {\rm Ebberup
  1986}}}, \bibinfo{publisher}{North-Holland}, \bibinfo{address}{Amsterdam},
  pp. \bibinfo{pages}{77--103}.

\bibitemdeclare{inproceedings}{BRV95}
\bibitem{BRV95}
\bibinfo{author}{E.~\surnamestart Brinksma\surnameend},
  \bibinfo{author}{A.~\surnamestart Rensink\surnameend} \&
  \bibinfo{author}{W.~\surnamestart Vogler\surnameend} (\bibinfo{year}{1995}):
  \emph{\bibinfo{title}{Fair Testing}}.
\newblock In \bibinfo{editor}{I.~\surnamestart Lee\surnameend} \&
  \bibinfo{editor}{S.~\surnamestart Smolka\surnameend}, editors: {\sl
  \bibinfo{booktitle}{{\rm Proceedings 6th International Conference on}
  Concurrency Theory, {\rm ({CONCUR}'95), Philadelphia, PA, USA, August
  1995}}}, {\sl \bibinfo{series}{\rm LNCS}} \bibinfo{volume}{962},
  \bibinfo{publisher}{Springer}, pp. \bibinfo{pages}{313--327},
  \doi{10.1007/3-540-60218-6_23}.

\bibitemdeclare{article}{CFG15}
\bibitem{CFG15}
\bibinfo{author}{T.~\surnamestart Chen\surnameend},
  \bibinfo{author}{W.J.~\surnamestart Fokkink\surnameend} \&
  \bibinfo{author}{R.J.~\surnamestart van Glabbeek\surnameend}
  (\bibinfo{year}{2015}): \emph{\bibinfo{title}{On the Axiomatizability of
  Impossible Futures}}.
\newblock {\sl \bibinfo{journal}{Logical Methods in Computer Science}}
  \bibinfo{volume}{11}(\bibinfo{number}{3}):\bibinfo{eid}{17},
  \doi{10.2168/LMCS-11(3:17)2015}.

\bibitemdeclare{article}{DH84}
\bibitem{DH84}
\bibinfo{author}{R.~\surnamestart De~Nicola\surnameend} \&
  \bibinfo{author}{M.~\surnamestart Hennessy\surnameend}
  (\bibinfo{year}{1984}): \emph{\bibinfo{title}{Testing equivalences for
  processes}}.
\newblock {\sl \bibinfo{journal}{Theoretical Computer Science}}
  \bibinfo{volume}{34}, pp. \bibinfo{pages}{83--133},
  \doi{10.1016/0304-3975(84)90113-0}.

\bibitemdeclare{article}{DGHM08}
\bibitem{DGHM08}
\bibinfo{author}{Y.~\surnamestart Deng\surnameend},
  \bibinfo{author}{R.J.~\surnamestart van Glabbeek\surnameend},
  \bibinfo{author}{M.~\surnamestart Hennessy\surnameend} \&
  \bibinfo{author}{C.C.~\surnamestart Morgan\surnameend}
  (\bibinfo{year}{2008}): \emph{\bibinfo{title}{Characterising Testing
  Preorders for Finite Probabilistic Processes}}.
\newblock {\sl \bibinfo{journal}{Logical Methods in Computer Science}}
  \bibinfo{volume}{4}(\bibinfo{number}{4}):\bibinfo{eid}{4},
  \doi{10.2168/LMCS-4(4:4)2008}.

\bibitemdeclare{article}{DGHM13}
\bibitem{DGHM13}
\bibinfo{author}{Y.~\surnamestart Deng\surnameend},
  \bibinfo{author}{R.J.~\surnamestart van Glabbeek\surnameend},
  \bibinfo{author}{M.~\surnamestart Hennessy\surnameend} \&
  \bibinfo{author}{C.C.~\surnamestart Morgan\surnameend}
  (\bibinfo{year}{2014}): \emph{\bibinfo{title}{Real-Reward Testing for
  Probabilistic Processes}}.
\newblock {\sl \bibinfo{journal}{Theoretical Computer Science}}
  \bibinfo{volume}{538}, pp. \bibinfo{pages}{16--36},
  \doi{10.1016/j.tcs.2013.07.016}.

\bibitemdeclare{incollection}{DGHMZ07}
\bibitem{DGHMZ07}
\bibinfo{author}{Y.~\surnamestart Deng\surnameend},
  \bibinfo{author}{R.J.~\surnamestart van Glabbeek\surnameend},
  \bibinfo{author}{M.~\surnamestart Hennessy\surnameend},
  \bibinfo{author}{C.C.~\surnamestart Morgan\surnameend} \&
  \bibinfo{author}{C.~\surnamestart Zhang\surnameend} (\bibinfo{year}{2007}):
  \emph{\bibinfo{title}{Remarks on Testing Probabilistic Processes}}.
\newblock In \bibinfo{editor}{L.~\surnamestart Cardelli\surnameend},
  \bibinfo{editor}{M.~\surnamestart Fiore\surnameend} \&
  \bibinfo{editor}{G.~\surnamestart Winskel\surnameend}, editors: {\sl
  \bibinfo{booktitle}{Computation, Meaning, and Logic: Articles dedicated to
  Gordon Plotkin}}, {\sl \bibinfo{series}{Electronic Notes in Theoretical
  Computer Science}} \bibinfo{volume}{172}, \bibinfo{publisher}{Elsevier}, pp.
  \bibinfo{pages}{359--397}, \doi{10.1016/j.entcs.2007.02.013}.

\bibitemdeclare{inproceedings}{DGMZ07}
\bibitem{DGMZ07}
\bibinfo{author}{Y.~\surnamestart Deng\surnameend},
  \bibinfo{author}{R.J.~\surnamestart van Glabbeek\surnameend},
  \bibinfo{author}{C.C.~\surnamestart Morgan\surnameend} \&
  \bibinfo{author}{C.~\surnamestart Zhang\surnameend} (\bibinfo{year}{2007}):
  \emph{\bibinfo{title}{Scalar Outcomes Suffice for Finitary Probabilistic
  Testing}}.
\newblock In \bibinfo{editor}{R.~\surnamestart {De Nicola}\surnameend}, editor:
  {\sl \bibinfo{booktitle}{{\rm Proceedings 16th} European Symposium on
  Programming, {\rm ESOP 2007, Braga, Portugal}}}, {\sl \bibinfo{series}{\rm
  LNCS}} \bibinfo{volume}{4421}, \bibinfo{publisher}{Springer}, pp.
  \bibinfo{pages}{363--378}, \doi{10.1007/978-3-540-71316-6\_25}.

\bibitemdeclare{misc}{vG09}
\bibitem{vG09}
\bibinfo{author}{R.J.~\surnamestart van Glabbeek\surnameend}
  (\bibinfo{year}{2009}): \emph{\bibinfo{title}{The Linear Time – Branching
  Time Spectrum after 20 years, {\rm or} Full abstraction for safety and
  liveness properties}}.
\newblock \bibinfo{howpublished}{Copies of slides. Invited talk for IFIP WG 1.8
  at CONCUR 2009 in Bologna}.
\newblock
  \urlprefix\url{http://theory.stanford.edu/~rvg/abstracts.html#20years}.

\bibitemdeclare{inproceedings}{vG10}
\bibitem{vG10}
\bibinfo{author}{R.J.~\surnamestart van Glabbeek\surnameend}
  (\bibinfo{year}{2010}): \emph{\bibinfo{title}{The Coarsest Precongruences
  Respecting Safety and Liveness Properties}}.
\newblock In \bibinfo{editor}{C.S.~\surnamestart Calude\surnameend} \&
  \bibinfo{editor}{V.~\surnamestart Sassone\surnameend}, editors: {\sl
  \bibinfo{booktitle}{{\rm Proceedings 6th IFIP TC 1/WG 2.2 International
  Conference on} Theoretical Computer Science {\rm (TCS 2010); held as part of
  the {\sl World Computer Congress} 2010, Brisbane, Australia}}}, {\sl
  \bibinfo{series}{IFIP}} \bibinfo{volume}{323}, \bibinfo{publisher}{Springer},
  pp. \bibinfo{pages}{32--52}, \doi{10.1007/978-3-642-15240-5_3}.

\bibitemdeclare{misc}{vG16}
\bibitem{vG16}
\bibinfo{author}{R.J.~\surnamestart van Glabbeek\surnameend}
  (\bibinfo{year}{2016}): \emph{\bibinfo{title}{Ensuring Liveness Properties of
  Distributed Systems (A Research Agenda)}}.
\newblock \bibinfo{howpublished}{Position paper}.
\newblock \urlprefix\url{https://arxiv.org/abs/1711.04240}.

\bibitemdeclare{inproceedings}{vG17b}
\bibitem{vG17b}
\bibinfo{author}{R.J.~\surnamestart van Glabbeek\surnameend}
  (\bibinfo{year}{2017}): \emph{\bibinfo{title}{Lean and Full Congruence
  Formats for Recursion}}.
\newblock In: {\sl \bibinfo{booktitle}{{\rm Proceedings $32^{nd}$ Annual
  ACM/IEEE Symposium on} Logic in Computer Science, {\rm LICS 2017, Reykjavik,
  Iceland, June 2017}}}, \bibinfo{publisher}{IEEE Computer Society Press},
  \doi{10.1109/LICS.2017.8005142}.

\bibitemdeclare{techreport}{GH15a}
\bibitem{GH15a}
\bibinfo{author}{R.J.~\surnamestart van Glabbeek\surnameend} \&
  \bibinfo{author}{P.~\surnamestart H{\"o}fner\surnameend}
  (\bibinfo{year}{2015}): \emph{\bibinfo{title}{Progress, Fairness and Justness
  in Process Algebra}}.
\newblock \bibinfo{type}{Technical Report} \bibinfo{number}{8501},
  \bibinfo{institution}{NICTA}, \bibinfo{address}{Sydney, Australia}.
\newblock \urlprefix\url{http://arxiv.org/abs/1501.03268}.

\bibitemdeclare{inproceedings}{Hen82}
\bibitem{Hen82}
\bibinfo{author}{M.~\surnamestart Hennessy\surnameend} (\bibinfo{year}{1982}):
  \emph{\bibinfo{title}{Powerdomains and nondeterministic recursive
  definitions}}.
\newblock In: {\sl \bibinfo{booktitle}{{\rm Proceedings 5th Intern.} Symposium
  on Programming}}, {\sl \bibinfo{series}{\rm LNCS}} \bibinfo{volume}{137},
  \bibinfo{publisher}{Springer}, pp. \bibinfo{pages}{178--193},
  \doi{10.1007/3-540-11494-7\_13}.

\bibitemdeclare{book}{henn}
\bibitem{henn}
\bibinfo{author}{M.~\surnamestart Hennessy\surnameend} (\bibinfo{year}{1988}):
  \emph{\bibinfo{title}{An Algebraic Theory of Processes}}.
\newblock \bibinfo{publisher}{MIT Press}.

\bibitemdeclare{inproceedings}{JHW94}
\bibitem{JHW94}
\bibinfo{author}{B.~\surnamestart Jonsson\surnameend},
  \bibinfo{author}{C.~\surnamestart Ho-Stuart\surnameend} \&
  \bibinfo{author}{W.~\surnamestart Yi\surnameend} (\bibinfo{year}{1994}):
  \emph{\bibinfo{title}{Testing and Refinement for Nondeterministic and
  Probabilistic Processes}}.
\newblock In: {\sl \bibinfo{booktitle}{{\rm Proceedings of the 3rd
  International Symposium on} Formal Techniques in Real-Time and Fault-Tolerant
  Systems}}, {\sl \bibinfo{series}{\rm LNCS}} \bibinfo{volume}{863},
  \bibinfo{publisher}{Springer}, pp. \bibinfo{pages}{418--430},
  \doi{10.1007/3-540-58468-4_176}.

\bibitemdeclare{inproceedings}{KV92}
\bibitem{KV92}
\bibinfo{author}{R.~\surnamestart Kaivola\surnameend} \&
  \bibinfo{author}{A.~\surnamestart Valmari\surnameend} (\bibinfo{year}{1992}):
  \emph{\bibinfo{title}{The Weakest Compositional Semantic Equivalence
  Preserving Nexttime-less Linear Temporal Logic}}.
\newblock In \bibinfo{editor}{R.~\surnamestart Cleaveland\surnameend}, editor:
  {\sl \bibinfo{booktitle}{CONCUR'92}}, {\sl \bibinfo{series}{\rm LNCS}}
  \bibinfo{volume}{630}, \bibinfo{publisher}{Springer}, pp.
  \bibinfo{pages}{207--221}, \doi{10.1007/BFb0084793}.

\bibitemdeclare{article}{Lam77}
\bibitem{Lam77}
\bibinfo{author}{L.~\surnamestart Lamport\surnameend} (\bibinfo{year}{1977}):
  \emph{\bibinfo{title}{Proving the correctness of multiprocess programs}}.
\newblock {\sl \bibinfo{journal}{IEEE Transactions on Software Engineering}}
  \bibinfo{volume}{3}(\bibinfo{number}{2}), pp. \bibinfo{pages}{125--143},
  \doi{10.1109/TSE.1977.229904}.

\bibitemdeclare{article}{Lam98}
\bibitem{Lam98}
\bibinfo{author}{L.~\surnamestart Lamport\surnameend} (\bibinfo{year}{1998}):
  \emph{\bibinfo{title}{Proving Possibility Properties}}.
\newblock {\sl \bibinfo{journal}{Theoretical Computer Science}}
  \bibinfo{volume}{206}(\bibinfo{number}{1-2}), pp. \bibinfo{pages}{341--352},
  \doi{10.1016/S0304-3975(98)00129-7}.
\newblock \bibinfo{note}{See especially
  \url{http://research.microsoft.com/en-us/um/people/lamport/pubs/pubs.html\#lamport-possibility}}.

\bibitemdeclare{inproceedings}{Le94}
\bibitem{Le94}
\bibinfo{author}{G.~\surnamestart Leduc\surnameend} (\bibinfo{year}{1994}):
  \emph{\bibinfo{title}{Failure-based congruences, unfair divergences and new
  testing theory}}.
\newblock In \bibinfo{editor}{S.T.~\surnamestart Vuong\surnameend} \&
  \bibinfo{editor}{S.T.~\surnamestart Chanson\surnameend}, editors: {\sl
  \bibinfo{booktitle}{{\rm Proceedings Fourteenth {IFIP} {WG6.1} International
  Symposium on} Protocol Specification, Testing and Verification, {\rm
  Vancouver, BC, Canada, 1994}}}, {\sl \bibinfo{series}{{IFIP} Conference
  Proceedings}}~\bibinfo{volume}{1}, \bibinfo{publisher}{Chapman {\&} Hall},
  pp. \bibinfo{pages}{252--267}.

\bibitemdeclare{article}{Levy08}
\bibitem{Levy08}
\bibinfo{author}{P.B.~\surnamestart Levy\surnameend} (\bibinfo{year}{2008}):
  \emph{\bibinfo{title}{Infinite trace equivalence}}.
\newblock {\sl \bibinfo{journal}{Annals of Pure and Applied Logic}}
  \bibinfo{volume}{151}(\bibinfo{number}{2-3}), pp. \bibinfo{pages}{170--198},
  \doi{10.1016/j.apal.2007.10.007}.

\bibitemdeclare{incollection}{ccs}
\bibitem{ccs}
\bibinfo{author}{R.~\surnamestart Milner\surnameend} (\bibinfo{year}{1990}):
  \emph{\bibinfo{title}{Operational and algebraic semantics of concurrent
  processes}}.
\newblock In \bibinfo{editor}{J.~\surnamestart van Leeuwen\surnameend}, editor:
  {\sl \bibinfo{booktitle}{Handbook of Theoretical Computer Science}},
  chapter~\bibinfo{chapter}{19}, \bibinfo{publisher}{Elsevier Science
  Publishers B.V. (North-Holland)}, pp. \bibinfo{pages}{1201--1242}.
\newblock \bibinfo{note}{Alternatively see{ \em Communication and Concurrency},
  Prentice-Hall, Englewood Cliffs, 1989}.

\bibitemdeclare{inproceedings}{NC95}
\bibitem{NC95}
\bibinfo{author}{V.~\surnamestart Natarajan\surnameend} \&
  \bibinfo{author}{R.~\surnamestart Cleaveland\surnameend}
  (\bibinfo{year}{1995}): \emph{\bibinfo{title}{Divergence and Fair Testing}}.
\newblock In \bibinfo{editor}{Z.~\surnamestart F{\"{u}}l{\"{o}}p\surnameend} \&
  \bibinfo{editor}{F.~\surnamestart G{\'{e}}cseg\surnameend}, editors: {\sl
  \bibinfo{booktitle}{{\rm Proceedings 22nd International Colloquium on}
  Automata, Languages and Programming {\rm (ICALP'95), Szeged, Hungary, July
  1995}}}, {\sl \bibinfo{series}{\rm LNCS}} \bibinfo{volume}{944},
  \bibinfo{publisher}{Springer}, pp. \bibinfo{pages}{648--659},
  \doi{10.1007/3-540-60084-1\_112}.

\bibitemdeclare{article}{OH86}
\bibitem{OH86}
\bibinfo{author}{E.-R.~\surnamestart Olderog\surnameend} \&
  \bibinfo{author}{C.A.R.~\surnamestart Hoare\surnameend}
  (\bibinfo{year}{1986}): \emph{\bibinfo{title}{Specification-oriented
  semantics for communicating processes}}.
\newblock {\sl \bibinfo{journal}{Acta Informatica}} \bibinfo{volume}{23}, pp.
  \bibinfo{pages}{9--66}, \doi{10.1007/BF00268075}.

\bibitemdeclare{inproceedings}{Pu01}
\bibitem{Pu01}
\bibinfo{author}{A.~\surnamestart Puhakka\surnameend} (\bibinfo{year}{2001}):
  \emph{\bibinfo{title}{Weakest Congruence Results Concerning ``Any-Lock''}}.
\newblock In \bibinfo{editor}{N.~\surnamestart Kobayashi\surnameend} \&
  \bibinfo{editor}{B.~\surnamestart Pierce\surnameend}, editors: {\sl
  \bibinfo{booktitle}{{\rm Proceedings 4th International Symposium on}
  Theoretical Aspects of Computer Software, {\rm TACS 2001, Sendai, Japan,
  2001}}}, {\sl \bibinfo{series}{\rm LNCS}} \bibinfo{volume}{2215},
  \bibinfo{publisher}{Springer}, pp. \bibinfo{pages}{400--419},
  \doi{10.1007/3-540-45500-0\_20}.

\bibitemdeclare{book}{Ros97}
\bibitem{Ros97}
\bibinfo{author}{A.W.~\surnamestart Roscoe\surnameend} (\bibinfo{year}{1997}):
  \emph{\bibinfo{title}{The Theory and Practice of Concurrency}}.
\newblock \bibinfo{publisher}{Prentice-Hall}.
\newblock
  \urlprefix\url{http://www.comlab.ox.ac.uk/bill.roscoe/publications/68b.pdf}.

\bibitemdeclare{inproceedings}{Ros05}
\bibitem{Ros05}
\bibinfo{author}{A.W.~\surnamestart Roscoe\surnameend} (\bibinfo{year}{2005}):
  \emph{\bibinfo{title}{Seeing Beyond Divergence}}.
\newblock In \bibinfo{editor}{A.E.~\surnamestart Abdallah\surnameend},
  \bibinfo{editor}{C.B.~\surnamestart Jones\surnameend} \&
  \bibinfo{editor}{J.W.~\surnamestart Sanders\surnameend}, editors: {\sl
  \bibinfo{booktitle}{Communicating Sequential Processes: The First 25 Years,
  Symposium on the Occasion of 25 Years of CSP, {\rm London, UK, 2004, Revised
  Invited Papers}}}, {\sl \bibinfo{series}{\rm LNCS}} \bibinfo{volume}{3525},
  \bibinfo{publisher}{Springer}, pp. \bibinfo{pages}{15--35},
  \doi{10.1007/11423348\_2}.

\bibitemdeclare{inproceedings}{Seg96}
\bibitem{Seg96}
\bibinfo{author}{R.~\surnamestart Segala\surnameend} (\bibinfo{year}{1996}):
  \emph{\bibinfo{title}{Testing Probabilistic Automata}}.
\newblock In: {\sl \bibinfo{booktitle}{{\rm Proceedings of the 7th
  International Conference on} Concurrency Theory}}, {\sl \bibinfo{series}{\rm
  LNCS}} \bibinfo{volume}{1119}, \bibinfo{publisher}{Springer}, pp.
  \bibinfo{pages}{299--314}, \doi{10.1007/3-540-61604-7\_62}.

\bibitemdeclare{article}{Vo02}
\bibitem{Vo02}
\bibinfo{author}{W.~\surnamestart Vogler\surnameend} (\bibinfo{year}{2002}):
  \emph{\bibinfo{title}{Efficiency of asynchronous systems, read arcs, and the
  {MUTEX}-problem}}.
\newblock {\sl \bibinfo{journal}{Theoretical Computer Science}}
  \bibinfo{volume}{275}(\bibinfo{number}{1-2}), pp. \bibinfo{pages}{589--631},
  \doi{10.1016/S0304-3975(01)00300-0}.

\end{thebibliography}
\end{document}